%% file: main.tex
\documentclass{article}
\usepackage[utf8]{inputenc}
\usepackage[english]{babel}
\usepackage[T1]{fontenc}
\usepackage{csquotes}
\usepackage{amsfonts,amsthm,amsmath,amssymb,mathtools}
\usepackage{thmtools}
\usepackage{thm-restate}
\usepackage[defaultsups]{newtxtext}
\usepackage{dsfont,microtype}
\usepackage{graphicx, quantikz, tikz,tikz-3dplot}
\usepackage{mathrsfs}
\usepackage{cancel}
\usepackage{slashed}
\usepackage{booktabs,multirow}
\usepackage[section]{placeins}
\usepackage{float}
\usepackage{comment}
\usepackage{fullpage}
\usepackage{enumerate}
\usepackage{physics}
\usepackage{authblk}
\usepackage{svg}
\usepackage{yfonts}

\renewcommand{\norm}[1]{\left\lVert#1\right\rVert}
\newcommand{\AN}{\mathcal{A}_N}

\newcommand{\barAN}{\bar{\mathcal{A}}_N}
\newcommand{\PSPACE}{\mathsf{PSPACE}}
\newcommand{\EXPSPACE}{\mathsf{EXPSPACE}}
\newcommand{\Pclass}{\mathsf{P}}
\newcommand{\NP}{\mathsf{NP}}
\newcommand{\BQP}{\mathsf{BQP}}
\newcommand{\BQPSPACE}{\mathsf{BQPSPACE}}
\newcommand{\PrQSPACE}{\mathsf{PrQSPACE}}
\newcommand{\FSTherm}{\mathsf{FSTherm (MC)}}
\newcommand{\FSThermG}{\mathsf{FSTherm (Gibbs)}}
\newcommand{\FSRelax}{\mathsf{FSRelax}}
\newcommand{\FTFSRelax}{\mathsf{FTFSRelax}}
\newcommand{\FSHalt}{\mathsf{FSHalt}}

\newcommand{\Heff}{H_{\text{eff}}}
\newcommand{\rhoMC}{\rho_{\textsf{MC}}}
\newcommand{\rhop}{\rho_{\mathrm{poly}}}
\newcommand{\psip}{\psi_{\mathrm{poly}}}
\newcommand{\rhog}{\rho_{\beta}}
\newcommand{\rhogp}{\rho_{\beta'}}
\newcommand{\pgauss}[2]{P_{\mathrm{Gauss}(#1)}^{(#2)}}

\newcommand{\LocalDim}{51 }

\DeclareMathOperator{\traceOp}{tr}
\DeclareMathOperator{\polyOp}{poly}
\DeclareMathOperator{\logOp}{log}
\DeclarePairedDelimiterXPP\bigO[1]{\ensuremath{\mathcal{O}}}{(}{)}{}{#1}
\DeclarePairedDelimiterXPP\bigomega[1]{\ensuremath{\Omega}}{(}{)}{}{#1}
\DeclarePairedDelimiterXPP\bigtheta[1]{\ensuremath{\Theta}}{(}{)}{}{#1}
\DeclarePairedDelimiterXPP\trdummy[1]{\traceOp}{[}{]}{}{#1}
\DeclarePairedDelimiterXPP\ptr[2]{\traceOp_{#1}}{[}{]}{}{#2}
\DeclarePairedDelimiterXPP\poly[1]{\polyOp}{(}{)}{}{#1}
\DeclarePairedDelimiterXPP\logdummy[1]{\ensuremath{\logOp}}{(}{)}{}{#1}
\renewcommand{\tr}[1]{\trdummy{#1}}
\renewcommand{\log}[1]{\logdummy{#1}}
\definecolor{LightPink}{rgb}{0.858, 0.188, 0.478}
\usepackage[colorlinks=true, urlcolor=blue,citecolor=purple,anchorcolor=blue, linkcolor=LightPink]{hyperref}

\definecolor{navyblue}{rgb}{0.0, 0.0, 0.5}

\usepackage[backend=biber,doi=false,isbn=false,url=false,maxbibnames=5,maxcitenames=2,sorting=none]{biblatex}
\DeclareBibliographyDriver{inproceedings}{%
  \usebibmacro{bibindex}%
  \usebibmacro{begentry}%
  \usebibmacro{author/translator+others}%
  \setunit{\labelnamepunct}\newblock
  \iffieldundef{doi}
    {\printfield[title]{title}}
    {\href{https://doi.org/\thefield{doi}}{\printfield[title]{title}}}%
  \newunit\newblock
  \printfield{booktitle}%
  \newunit\newblock
  \printlist{publisher}%
  \printfield{year}%
  \newunit\newblock
  \printfield{pages}%
  \newunit\newblock
  \usebibmacro{doi+eprint+url}%
  \usebibmacro{finentry}%
}
\DeclareBibliographyDriver{inbook}{%
  \usebibmacro{bibindex}%
  \usebibmacro{begentry}%
  \usebibmacro{author/translator+others}%
  \setunit{\labelnamepunct}\newblock
  \iffieldundef{doi}
    {\printfield[title]{title}}
    {\href{https://doi.org/\thefield{doi}}{\printfield[title]{title}}}%
  \newunit\newblock
  \printfield{booktitle}%
  \newunit\newblock
  \printlist{publisher}%
  \printfield{year}%
  \newunit\newblock
  \printfield{pages}%
  \newunit\newblock
  \usebibmacro{doi+eprint+url}%
  \usebibmacro{finentry}%
}
\bibliography{references}
\newbibmacro{string+doi}[1]{\iffieldundef{doi}{#1}{\href{https://dx.doi.org/\thefield{doi}}{#1}}}
\DeclareFieldFormat{title}{\usebibmacro{string+doi}{\mkbibemph{#1}}}
\DeclareFieldFormat[article]{title}{\usebibmacro{string+doi}{\mkbibquote{#1}}}
\DeclareFieldFormat[incollection]{title}{\usebibmacro{string+doi}{\mkbibquote{#1}}}
\input{symphony.tex}

\title{The Complexity of Thermalization in Finite Quantum Systems}
\date{}

\author[1,3,4]{\href{https://orcid.org/0000-0002-2612-308X}{Dhruv~Devulapalli}}
\author[1,3,4]{\href{https://orcid.org/0000-0001-9727-6967}{T.~C.~Mooney}}
\author[1,2]{\href{https://orcid.org/0000-0002-6077-4898}{James~D.~Watson}}

\affil[1]{\normalsize  Joint Center for Quantum Information \& Computer Science, 
	National Institute of Standards \& Technology 
	\authorcr and University of Maryland, 
	College Park }
\affil[2]{\normalsize Department of Computer Science and Institute for Advanced Computer Studies,
	\authorcr 
	University of Maryland, College Park}
\affil[3]{\normalsize Joint Quantum Institute, NIST/University of Maryland, 
        College Park, MD 20742, USA}
\affil[4]{\normalsize Department of Physics,
	University of Maryland, College Park}

\begin{document}

	{\begingroup
		\hypersetup{urlcolor=navyblue}
		\maketitle
		\endgroup}

\begin{abstract}
Thermalization is the process through which a physical system evolves toward a state of thermal equilibrium. Determining whether or not a physical system will thermalize from an initial state has been a key question in condensed matter physics. Closely related questions are determining whether observables in these systems relax to stationary values, and what those values are.
    Using tools from computational complexity theory, we demonstrate that given a Hamiltonian on a finite-sized system, determining whether or not it thermalizes or relaxes to a given stationary value is computationally intractable, even for a quantum computer.
    In particular, we show that the problem of determining whether an observable of a finite-sized quantum system relaxes to a given value is $\PSPACE$-complete, and so no efficient algorithm for determining the value is expected to exist. Further, we show the existence of Hamiltonians for which the problem of determining whether the system thermalizes to the Gibbs expectation value is $\PSPACE$-complete. We also show that the related problem of determining whether the system thermalizes to the microcanonical expectation value is contained in $\PSPACE$ and is $\PSPACE$-hard under quantum polynomial time reductions. 
    In light of recent results 
    demonstrating undecidability of thermalization in the thermodynamic limit, our work shows that the intractability of the problem is due to inherent difficulties in many-body physics rather than particularities of infinite systems.
\end{abstract}

\section{Introduction}
Statistical mechanics has been remarkably successful at describing the physics of large numbers of particles in thermal equilibrium. Experimental outcomes of systems in thermal equilibrium can be predicted using statistical ensembles based on macroscopic observables, independently of the system's microstate. Therefore, the process through which an isolated system evolves toward thermal equilibrium,  known as thermalization, has been a central topic of investigation in condensed matter physics.
In classical systems, the emergence of thermalization has been attributed to ergodicity and chaotic dynamics  \cite{Deutsch_2018}. Ergodic systems evolve such that they eventually come arbitrarily close to every point in phase space on a constant energy surface.
As classical measurements can be modeled as long-time averages, the long-time behavior of ergodic systems is hence similar to that of a uniform distribution over the constant-energy submanifold of phase space, also known as the classical microcanonical ensemble \cite{Pathria2021,Ruelle1999}. The true relationship between thermalization and ergodicity in classical systems is subtle and nuanced. Ergodicity is difficult to prove and has only been rigorously shown for some systems \cite{Sinai1970,Bunimovich1979,Simanyi2004,Chernov2006}, notably for systems with suitably strong chaotic behavior \cite{Ott2002,deAlmeida1989}. Indeed, it has been argued that ergodicity is neither necessary nor sufficient for thermalization \cite{Bricmont_1995}. Nevertheless, ergodicity has remained a useful conceptual tool to understand thermalization in classical systems.

However, the ergodic explanation of thermalization can't be naturally extended to isolated quantum systems. In particular, unitary time evolution does not sample all wavefunctions within a given energy range of the starting state, and hence cannot be ergodic. One explanation for thermalization in quantum systems is in terms of \emph{subsystem} equilibration. \textcite{linden_quantum_2009, Short_Farrelly_2012} have shown that subsystems equilibrate under some assumptions on their Hamiltonian. Related work \cite{Reimann_2008} has shown that experimentally realistic conditions can lead to equilibration, based on high density of states in physically realistic scenarios. An alternative explanation, the Eigenstate Thermalization Hypothesis (ETH) \cite{Srednicki_1994, Deutsch_1991,Deutsch_2018} has been proposed to explain the exhibition of thermalization in isolated quantum systems. ETH proposes that observables with slowly varying diagonal elements (in the energy eigenbasis) and small off-diagonal elements approach their microcanonical thermal values. However, the ETH has not been proven to be generally applicable. Indeed, systems exist that do not thermalize, due to phenomena such as many-body localization \cite{Abanin_Altman_Bloch_Serbyn_2019}, quantum many-body scars \cite{turner_scars}, and integrability \cite{Deutsch_2018,Kinoshita_Wenger_Weiss_2006} (although one can consider integrable systems as equilibrating to generalized Gibbs ensembles that take their additional conservation laws into account \cite{Langen2015,Vidmar2016}).

A large volume of theoretical and experimental work has gone into determining which properties and initial states ensure thermalization \cite{Tasaki_2016,Rigol_Dunjko_Olshanii_2008, Kim_Ikeda_Huse_2014, Gogolin_Eisert_2016, roadmap_2025}.
A natural question thus arises: How can we determine whether a given quantum system thermalizes? 
While this question has been studied by physicists, it was first considered from the viewpoint of computational complexity theory by \textcite{shiraishithermal2021}, who proved that even for product initial states with translationally invariant Hamiltonians, the problem of determining whether a system thermalizes or not is undecidable in the limit of an infinitely large system.
Their results utilize techniques inspired by previous work in the field of Hamiltonian complexity proving complexity and computability results in physics \cite{Cubitt_Perez-Garcia_Wolf_2015, Moore_1990}.

However, in the lab --- and nature more generally --- we only have access to physical systems of finite size. This raises the question of how difficult it is to predict thermalization of these finite-sized systems.
While the result of \textcite{shiraishithermal2021} proving undecidability in the thermodynamic limit, along with the more general hardness of simulating quantum systems, suggests this may be a difficult task, it is entirely possible that there exists an efficient algorithm for it. The result of \textcite{shiraishithermal2021} makes use of a reduction from an undecidable problem (the Halting problem) to the problem of determining whether a system thermalizes, but undecidable instances may require a large (unbounded) size and therefore may not be seen in systems of finite size such as those encountered in nature and the lab.

In this work we consider the problem of determining whether a \textit{finite-sized} quantum system relaxes to a given value, or thermalizes. 
To study this from a complexity-theoretic perspective, we formulate relaxation as a decision problem. 
\begin{definition}[(Informal) Finite-Size Relaxation, $\FSRelax$]
\label{def:fsrelax_informal}
    Consider an input of a $k$-local Hamiltonian $H$ acting on $N$ qudits with constant local Hilbert space dimension $d$, a sum of local  observables $\AN= \frac{1}{N}\sum_{i}A_i$, an input classical description of product state $\ket{\psi_0}$, a value $A^*$, and constants $c, \epsilon$.
    Let the long time average of $\AN$ be
    \begin{align}
        \barAN \coloneqq \lim_{T\rightarrow\infty}\frac{1}{T}\int_0^Tdt   \bra{\psi_0}e^{iHt}\AN e^{-iHt}\ket{\psi_0}
    \end{align}
    Then, output: YES if $|\barAN - A^*|\leq \epsilon$ or NO if $|\barAN - A^*|\geq c\epsilon$.
\end{definition}
$\FSRelax$ is depicted in \cref{fig:decision_problem}.
\begin{figure}
    \centering
\includegraphics[width=0.9\linewidth]{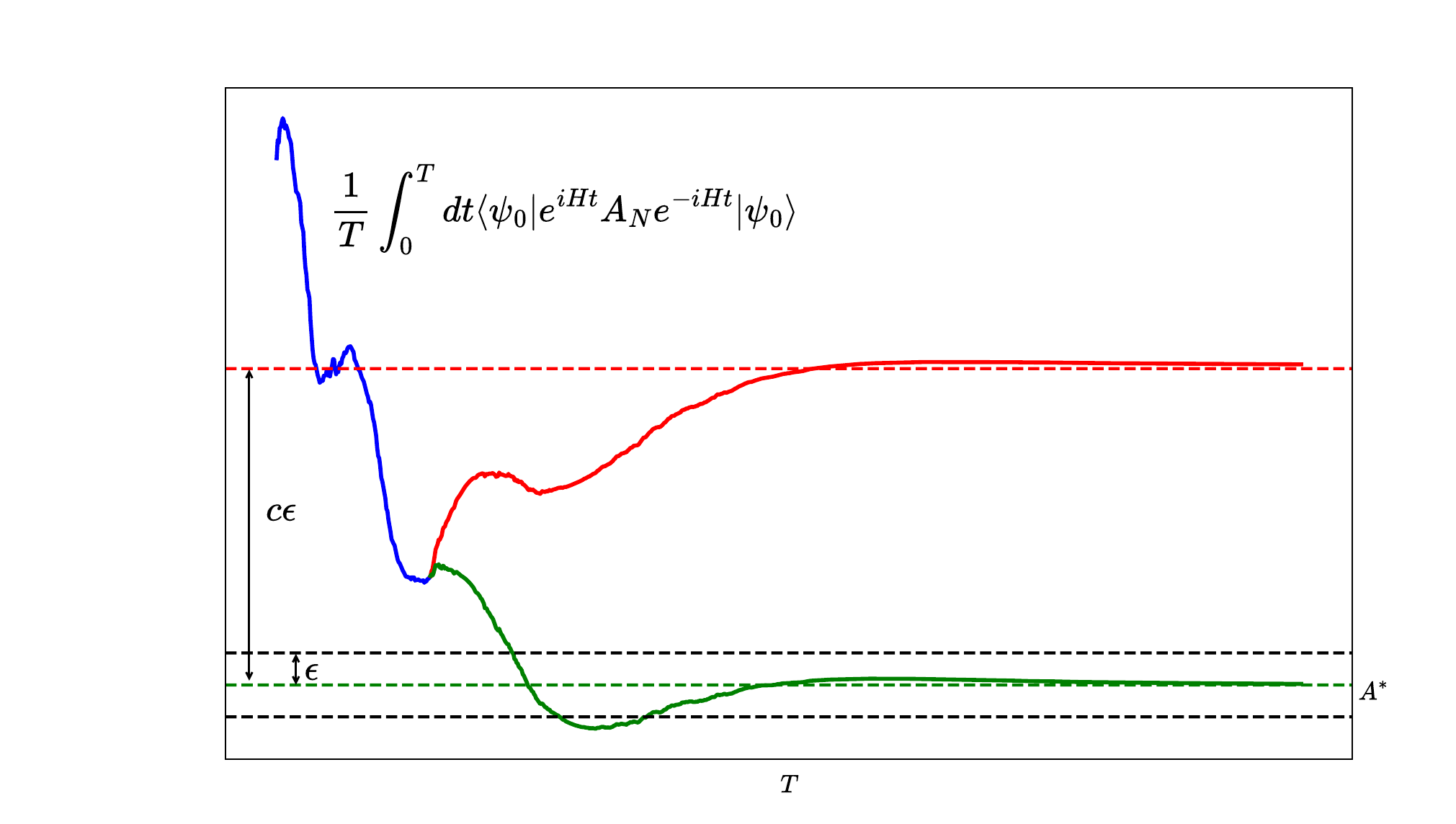}
    \caption{The relaxation decision problem. For a given initial state $\ket{\psi}$, we plot the time-average of $\AN \coloneqq \frac{1}{T}\int_0^Tdt   \bra{\psi_0}e^{iHt}\AN e^{-iHt}\ket{\psi_0}$. $A^*$ is indicated by the dashed green line. Given the promise that $\barAN\coloneqq\lim_{T\rightarrow\infty} \langle \AN \rangle_T$ is such that either $|\barAN - A^*| \geq c\epsilon$ or $|\barAN - A^*| \leq \epsilon, \FSRelax$ is the problem of deciding which of these two holds.}  
    \label{fig:decision_problem}
\end{figure}

One relevant state of interest that describes the thermal behavior of an isolated system is the microcanonical ensemble, which represents the set of possible states within a fixed range of energy of our initial state.
\begin{definition}[Microcanonical Ensemble]
\label{def:mc} For a Hamiltonian $H$ with eigenstates $\{\ket{\lambda_i}\}$ and corresponding eigenenergies $\{\lambda_i\}$ the quantum microcanonical ensemble $\rhoMC(H,E)$ at energy $E$ is defined as
    \begin{equation}
    \rhoMC(H,E) = \frac{1}{W} \sum_i f\left(\frac{\lambda_i - E}{w}\right) \ketbra{\psi_i}{\psi_i},
\end{equation}
where $f$ is an energy window function (typically decaying rapidly in its argument), $W = \sum_i f\left(\frac{\lambda_i - E}{w}\right)$ and $w$ is a characteristic energy width.
\end{definition}

\noindent To investigate thermalization, we define the following decision problem:
\begin{definition}[(Informal) Finite-Size Thermalization to the microcanonical state, $\FSTherm$]
\label{def:fstherm_informal}
    Consider an input of a $k$-local Hamiltonian $H$ acting on $N$ qudits with local Hilbert space dimension $d$, a sum of local  observables $\AN= \frac{1}{N}\sum_{i}A_i$, an input classical description of product state $\ket{\psi_0}$, and constants $c, \epsilon$.
    Let the long time average of $\AN$ be
    \begin{align}
        \barAN \coloneqq \lim_{T\rightarrow\infty}\frac{1}{T}\int_0^Tdt   \bra{\psi_0}e^{iHt}\AN e^{-iHt}\ket{\psi_0}.
    \end{align}
    
    Then, output: YES if $|\barAN - \tr{\AN \rhoMC(E)}|\leq \epsilon$ or NO if $|\barAN - \tr{\AN \rhoMC(E)}|\geq c\epsilon$, where $E$ is the expected energy of the input state $\bra{\psi_0} H \ket{\psi_0}$.
\end{definition}
We also consider a closely related problem to $\FSTherm$ which is the problem of thermalization to the Gibbs state $(\frac{e^{-\beta H}}{\tr{e^{-\beta H}}})$.
\begin{definition}[(Informal) Finite-Size Thermalization to the Gibbs state, $\FSThermG$]
\label{def:fsthermg_informal}
    Consider an input of a $k$-local Hamiltonian $H$ acting on $N$ qudits with local Hilbert space dimension $d$, a sum of local  observables $\AN= \frac{1}{N}\sum_{i}A_i$, an input classical description of product state $\ket{\psi_0}$, and constants $c, \epsilon$.
    Let the long time average of $\AN$ be
    \begin{align}
        \barAN \coloneqq \lim_{T\rightarrow\infty}\frac{1}{T}\int_0^Tdt   \bra{\psi_0}e^{iHt}\AN e^{-iHt}\ket{\psi_0}.
    \end{align}
    
    Then, output: YES if $|\barAN - \tr{\AN \rhog}|\leq \epsilon$ or NO if $|\barAN - \tr{\AN \rhog}|\geq c\epsilon$, where where $\rhog$ is the Gibbs state of $H$ with expected energy $E\coloneqq \bra{\psi_0} H \ket{\psi_0}$.
\end{definition}
Our main results are that $\FSRelax, \FSTherm,$ and $\FSThermG$ are computationally intractable, even for a quantum computer and for simple systems. We first prove that the problem of determining whether a given observable relaxes to a particular value is $\PSPACE$-complete.
\begin{restatable}{theorem}{fsrelaxpspacecomplete}
\label{thm:fsrelax_pspace_complete}
    $\FSRelax$ is PSPACE-complete for a translationally invariant, nearest neighbour Hamiltonian with $d\geq \LocalDim$, with a 1-local translationally invariant observable $A_i$.
\end{restatable}

We then show a $\PSPACE$ algorithm to compute the microcanonical expectation value of a local observable, which allows us to make use of the $\PSPACE$ algorithm to determine $\FSRelax$ to decide $\FSTherm$ in $\PSPACE$.
We also show that the $\FSTherm$ problem is contained in $\PSPACE$. By allowing a slightly larger local Hilbert space dimension, one can choose $A_i$ such that $\barAN = \tr{\rhoMC\AN}$, where $\rhoMC$ is the microcanonical ensemble of $H$. Therefore, we have:
\begin{restatable}{theorem}{fsthermpspacecomplete}
    \label{thm:fstherm_pspace_complete}
    $\FSTherm$ is $\PSPACE$-hard under quantum polynomial time reductions, and contained in $\PSPACE$.
\end{restatable}
Furthermore, our hardness result applies even for a translationally invariant, nearest neighbour Hamiltonian and a  1-local, translationally invariant observable.
However, we note that our proof of $\PSPACE$-hardness under quantum polynomial time reductions has the technical limitation that we restrict to a class of microcanonical states with  width $w = \bigomega{\norm{H}/\sqrt{\log{N}}}$. Further details on this restriction and the possibility of removing it in future work are discussed in \cref{sec:changingquantum} and \cref{sec:discussion}.

In the case of thermalization to the Gibbs state, we use similar techniques to show that $\FSThermG$ is contained in $\PSPACE$. However, using results from the tensor-network literature \cite{Alhambra_Cirac_2021}, we are able to show that $\FSThermG$ is $\PSPACE$-hard under \emph{classical} polynomial-time reductions, and is hence $\PSPACE$-complete. 
\begin{restatable}{theorem}{fsthermgpspacecontained}
    \label{thm:fsthermg_pspace_contained}
    $\FSThermG$ is  $\PSPACE$-complete.
\end{restatable}
Finally, we show that the experimentally relevant problem of determining whether finite-sized systems of qubits thermalize in polynomial time is $\BQP$-complete.

Our results show evidence for the difficulty of understanding the mechanisms of relaxation and thermalization, and it's intractability from a complexity theoretic perspective even for systems in nature or the lab.
\cref{thm:fsrelax_pspace_complete} demonstrates that we should not expect there to be an efficient way of finding (or even approximating) the time-averaged expectation value of local observables.
Furthermore, \cref{thm:fstherm_pspace_complete} shows that in general any ``easy-to-check'' condition which ensures thermalization will not apply to all Hamiltonians.
The systems we design may be of further physical interest, since our results show that even simple systems (translationally invariant, 1D nearest-neighbour Hamiltonians and initial product states) can exhibit highly complex relaxation and thermalization dynamics.

The rest of this paper is organized as follows: 
In \cref{sec:preliminaries} we define relevant complexity classes and introduce the relevant mathematical quantities for thermalization.
In \cref{sec:hardness} we prove the $\PSPACE$-hardness of $\FSRelax$.
In \cref{sec:containment} we prove that $\FSRelax \in \PSPACE$. Combined with our result on $\PSPACE$-hardness, this shows that $\FSRelax$ is $\PSPACE$-complete. In \cref{sec:fstherm_containment}, we prove that $\FSTherm \in \PSPACE$ and $\FSThermG \in \PSPACE$.  In \cref{sec:FSThermHardness}, we derive the quantum polynomial time reduction from $\FSTherm$ to $\FSRelax,$ as well as the classical polynomial time reduction of $\FSThermG$ to $\FSRelax$. Together, these prove our main results \cref{thm:fstherm_pspace_complete} and \cref{thm:fsthermg_pspace_contained}. Finally, in \cref{sec:FTFSRelax}, we prove that deciding whether a finite size system relaxes in finite time is $\BQP$-complete. We conclude with a discussion of our results and  outlook for future work in \cref{sec:discussion}.

\section{Preliminaries}
\label{sec:preliminaries}
We briefly introduce classical and quantum Turing machines. A more in-depth treatment can be found in \textcite{arora2009computational} and \textcite{Watrous_1999}. Turing machines consist of a tape with cells onto which symbols can be written, as well as a tape head, which  both stores a state and points to a specific tape cell. The tape head is able to read the symbol on the cell, and based on the read symbol and the head's state, (potentially) both write to the cell and move the head to an adjacent cell, according to a transition function. A quantum Turing machine can be defined similarly, but with a quantum tape and the ability to perform quantum operations (Toffoli gates, Hadamard gates, phase-shift gates, or single-qubit measurements in the computational basis) on the quantum tape \cite{watrous2008quantum}.
A promise problem is a pair $(A_{yes}, A_{no})$ where $A_{yes}, A_{no} \subseteq \{0,1\}^*$ are sets of strings satisfying $A_{yes} \cap A_{no} = \varnothing$ \cite{watrous2008quantum}.
We are primarily interested in the following classes of promise problems.
The first class $\PSPACE$, is the class of decision problems solvable by a deterministic Turing Machine in polynomial space.
More formally:
\begin{definition}[$\PSPACE$]
    A promise problem $A = (A_{yes}, A_{no})$ is in $\PSPACE$ (polynomial space) if and only if there exists a deterministic Turing machine running in polynomial space that accepts every string $x \in A_{yes}$ and rejects every string $x \in A_{no}$. 
\end{definition}

\noindent $\PSPACE$  contains the commonly known classes $\Pclass, \BQP, \NP$, and it is believed that $\PSPACE$ is not contained in them (and thus $\PSPACE$-complete problems are believed to be harder to solve).
There is a corresponding class for quantum computers, $\BQPSPACE$, which is roughly the class of decision problems solvable by a quantum Turing Machine in polynomial space.

\begin{definition}[$\BQPSPACE$,\cite{watrous2008quantum}]
    A promise problem $A = (A_{yes}, A_{no})$ is in $\BQPSPACE$ (bounded-error quantum polynomial space) if and only if there exists a quantum Turing machine running in polynomial space that accepts every string $x \in A_{yes}$ with probability at least 2/3 and accepts every string $x \in A_{no}$ with probability at most 1/3.
\end{definition}
\noindent It is known that $\BQPSPACE= \PSPACE$ \cite{watrous2008quantum}. 

\begin{definition}[$\PrQSPACE$ \cite{watrous2008quantum}]
    A promise problem $A = (A_{yes}, A_{no})$ is in $\PrQSPACE$ (probabilistic quantum polynomial space) if and only if there exists a quantum Turing machine
    $M$ running in polynomial space that accepts every string $x \in A_{yes}$ with probability strictly greater than $1/2$ and accepts every string $x \in A_{no}$ with probability at most $1/2$.
\end{definition}
\noindent It is known that $\PrQSPACE = \PSPACE$ \cite{watrous2008quantum}. 

\subsection{Algorithmic primitives}
We make use of two algorithmic primitives: block encodings and quantum phase estimation.
Block encodings allow the implementation of arbitrary sub-normalized matrices as the upper left block of a unitary \cite{gilyen2018qsvt, Childs_Kothari_Somma_2017, Berry_Childs_Cleve_Kothari_Somma_2015, block_encodings, Low2019hamiltonian}. 
We restate Definition 43 from \cite{gilyen2018qsvt} here:
\begin{definition}[Block Encoding]
\label{def:block_encoding}
    Suppose that $A$ is an $s$-qubit operator, $\alpha, \epsilon \in \mathbb{R}_{>0}$, and $a \in \mathbb{N}$. Then an $(s+a)$-qubit unitary $U$ is an $(\alpha, a, \epsilon)$-block encoding of $A$ if and only if
    \begin{equation}
        \norm{A - \alpha((I\otimes\bra{0^a}) U (I\otimes\ket{0^a}) )} \leq \epsilon
    \end{equation}
\end{definition}
Let $s$ represent the $s$-qubit register, and $a$ represent the ancillas used by the block-encoding. Then, for an initial state $\ket{\psi}$,
\begin{equation}
\label{eq:block_encoding_state}
    U\ket{\psi} = A'_s \ket{\psi}\ket{0^a}_a + \ket{\phi^{\perp}}
\end{equation}
where $\bra{0^a}_a \ket{\phi^{\perp}} = 0$ and $\norm{A-A'_s} \leq \epsilon$. Usually, when we make use of a block-encoding, we will postselect on the ancillas being in the state $\ket{0^a}$, since this heralds the correct application of $A'_s$ to $\ket{\psi}$.

We now restate some standard results on the complexity of preparing block encodings of (functions of) matrices.
\begin{lemma}[Block Encoding Sparse Matrices (Lemma 48 of \cite{gilyen2018qsvt})]\label{lemma:sparseblockenc} Let $A\in M_{2^w}(\mathbb{C})$ be a $s$-sparse matrix with $\vert A_{ij}\vert\leq 1$ for all $i,j,$ and for which both the position of nonzero entries and these non-zero entries' values have efficient classical descriptions (i.e., can be computed in polynomial time classically). Then, there exists a quantum circuit implementing a $(s, w+3, \varepsilon)$ block encoding of $A$ with $\bigO*{\poly{w}+\poly{\log{\frac{s^2}{\varepsilon}}}}$ fundamental gates and $w+\bigO*{\poly{\log{\frac{s^2}{\varepsilon}}}}$ space.
\end{lemma}

\begin{lemma}[Polynomial Functions of Block Encodings (Thm. 56 of \cite{gilyen2018qsvt})] \label{lemma:polyblock} Let $U\in U(2^n)$ be a  $(\alpha,a,\varepsilon)$-block encoding of some Hermitian $A$ with a circuit that uses $\mathscr{A}_U$ ancillas, $\mathscr{C}_U$ fundamental gates, and whose circuit has a classical description computable in time $\mathscr{T}_U$. Further, let $p\in\mathbb{R}[x]$ be a polynomial of degree $d$ such that $\sup_{x\in[-1,1]}\vert p(x)\vert \leq 1/2.$ Then, for any $\delta\geq 0$ there exists a  circuit using $\bigO{d(a+\mathscr{C}_U)}$ fundamental gates and $\bigO{d\mathscr{A}_U}$ ancillas 
to produce a $(1,a+2,4d\sqrt{\varepsilon/\alpha}+\delta)$ encoding $\tilde{U}_\delta$ of $p(A/\alpha).$ A classical description of this circuit can be computed in time $\bigO{d\mathscr{T}_U+\poly{d}+\poly{\log{\delta^{-1}}}}.$
\end{lemma}

To make full use of this, we need some polynomial approximations of specific functions we wish to block encode. 
\begin{lemma}[Polynomial Approximation of Square Root]
\label{lemma:polyapprox_sqrt}
For all $\epsilon\leq 1/2,$ there exists an efficiently computable polynomial $P^{(\epsilon)}_{\text{sqrt}}(x)\in\mathbb{C}[x]$ of degree $\bigO{\log{\epsilon^{-1}}}$ such that \begin{equation}\sup_{x\in[-1,1]}\left\vert P_{\mathrm{sqrt}}^{(\epsilon)}(x)-\sqrt{1+x/2}
\right\vert \leq \epsilon.\end{equation}
\end{lemma}
\begin{proof} 
Note that for $x\in [-1,1],$ $\sqrt{1+\frac{x}{2}}= \sum_{n=0}^\infty \frac{1}{1-2n}\binom{2n}{n}\left(-\frac{x}{8}\right)^n.$ We have that $\sum_{n=0}^\infty \frac{1}{\vert 1-2n\vert } \binom{2n}{n} 2^{-3n}=2-2^{-1/2}\approx 1.293.$
 
       Thus, by \cite[Cor. 66]{gilyen2018qsvt}, we have that for any $0<\epsilon \leq 1/2< 1/(2-2^{-1/2}),$
       there exists a polynomial $P_{\text{sqrt}}^{(\epsilon)}(x)$ of degree $
        \bigO{\log{\epsilon^{-1}}}$ (in the language of the cited corollary, $B=2
        -2^{-1/2},$ $\delta = r=1/2,$ and $x_0=0$) that satisfies $\sup_{x\in [-1,1]}\left\vert \sqrt{1+x/2}
        -P_{\text{sqrt}}^{(\epsilon)}(x)
        \right\vert \leq \epsilon$. 
\end{proof}
\begin{lemma}[Polynomial Approximations of Gaussian]
\label{lemma:polyapprox_gaussian}
For any $0<\epsilon \leq 1,$ there exists an efficiently computable polynomial $P_{\text{Gauss}(w)}^{(\epsilon)}(x)\in\mathbb{C}[x]$ of degree $\bigO{\max\{w^{-1},\sqrt{\log{\epsilon^{-1}}}\}\times\sqrt{\log{\epsilon^{-1}}}}$ such that 
$\sup_{x\in[-1,1]}\vert P_{\mathrm{Gauss}(w)}^{(\epsilon)}-e^{-\frac{\pi x^2}{2w^2}}\vert \leq \epsilon.$
\end{lemma}
\begin{proof} Note that if $x\in [-1,1],$ $\frac{\pi }{2}\left(\frac{x}{w}\right)^2\in [0, \frac{\pi}{2w^2}].$ Further, note that by \cite[Thm. 4.1]{Sachdeva_Vishnoi_2014}, there exists a polynomial $r_{\epsilon,w}$ of degree $\bigO*{ \sqrt{\max\left\{\frac{\pi}{2w^2}, \log{\epsilon^{-1}} \right\}\log{\epsilon^{-1}}}}$ such that $\sup_{x\in [0,\frac{\pi}{2w^2}]}\vert e^{-x}-r_{\epsilon,w}(x)\vert \leq \epsilon.$ Thus, letting $P_{\text{Gauss}(w)}^{(\epsilon)}(x):=r_{\epsilon, w}\left(\frac{\pi x^2}{w^2}\right),$ we arrive at our result.
\end{proof}
Next, we state a result about the number of ancillas necessary to implement phase estimation to $m$ bits of accuracy with total error $\leq \epsilon$ \cite{kitaev2002classical, Cleve1998}. 
\begin{lemma}[Space Complexity of Phase Estimation]\label{lemma:phaseestspacecomp}
    Given an $n$-qubit state $\ket{\psi},$ an 
    $n$-qubit unitary $U\in U(2^{n})$ such that $U\ket\psi = e^{2\pi i\phi},$ $m\in \mathbb{N},$ and $\epsilon\in (0,1)$ there exists a quantum circuit taking as input $\ket{\psi}\ket{0^{m+1+\lceil2\mathrm{log}_2\epsilon^{-1}\rceil}}$
    and using no additional ancillae which outputs a state $\ket{\psi}\ket{\tilde{\phi}^{(m,\epsilon)}}$ such that $\Vert \Pi_{\phi,m}\ket{\tilde{\phi}^{(m,\epsilon)}}\Vert \leq \epsilon,$ where $\Pi_{\phi,m}:=\sum_{\substack{\mathbf{b}\in \{0,1\}^{m+1+\lceil2\mathrm{log}_2\epsilon^{-1}\rceil}
    \\ \vert 0.b_1b_2\cdots - \phi\vert > 2^{-(m+1)}}}\ket{\boldsymbol{b}}\bra{\boldsymbol{b}}.$ 
\end{lemma}
\noindent In addition, we will need the following corollary about the approximate orthogonality of estimates of distant eigenvalues.
\begin{corollary}\label{cor:QPEOrtho}
    Let $\phi_1,\phi_2$ be eigenphases of $U$ as above, such that $\vert \phi_1-\phi_2\vert > 2^{-m}$ for some $m\in \mathbb{N}.$ Then, if $\ket{\phi_1^{(m,\epsilon)}},\ket{\phi_2^{(m,\epsilon)}}$ are as defined above, 
    \begin{equation}
        \left\vert\left\langle \phi_1^{(m,\epsilon)}\Big\vert \phi_2^{(m,\epsilon)}\right\rangle\right\vert \leq 2\epsilon.
    \end{equation}
\end{corollary}

\subsection{Thermalization}
Throughout, we will restrict ourselves to local Hamiltonians of the form:
\begin{align}
    H=\sum_{i=1}^m h_i
\end{align}
where $m=\poly{N}$, $h_i$ acts on at most an $\bigO{1}$ number of qudits, and $\norm{h_i}=\bigO{1}.$

Let $A_i$ be a local (single-qudit) observable acting on qudit $i$, then define: 
\begin{align}\label{eq:ANdef}
    \AN = \frac{1}{N}\sum^N_{i=1} A_i.
\end{align}
We will further mainly consider translationally invariant $\barAN$, such that each $A_i = A$.
Given an initial state $\ket{\psi_0}$ and a Hamiltonian $H$, the long time-average of $\AN$ is:
\begin{align}\label{Eq:Time-Averaged_Observable}
    \barAN \coloneqq \lim_{T\rightarrow\infty}\frac{1}{T}\int_0^Tdt  \bra{\psi_0}e^{iHt}\AN e^{-iHt}\ket{\psi_0}
\end{align}

The problem we will be interested in is, roughly, given an initial state $\ket{\psi_0}$, what expectation value of $\AN$ does this state relax to after an infinite amount of time?
In order to study this from a complexity theoretic viewpoint, we define this as a decision problem. 
\begin{definition}[(Formal) Finite Size Relaxation $\FSRelax$]
\label{def:fsrelax}
Consider a $k$-local Hamiltonian $H$ acting on $N$ qudits with constant local Hilbert space dimension $d$, a local  observable $A$, and a fixed value $A^*$.
Let $c, \epsilon$ be positive constants.
    Let the long time average of $A$ be as given in \cref{Eq:Time-Averaged_Observable}.
    
    \emph{\textbf{Input:}} A classical description of product state $\ket{\psi_0}$.
    
    \emph{\textbf{Output:}}  YES if $|\barAN - A^*|\leq \epsilon$ or NO if $|\barAN - A^*|\geq c\epsilon$.
    
    \emph{\textbf{Promise:}} Either $|\barAN - A^*|\leq \epsilon$ or $|\barAN - A^*|\geq c\epsilon$. Let $\{\ket{\lambda_i} \}$ be the eigenstates of $H$. Then, the initial state $\ket{\psi_0} := \sum_i c_i \ket{\lambda_i}$ is only supported on energy eigenstates with an inverse exponential gap, i.e. $\exists \, m = \bigO{\mathrm{poly}(N)} \, s.t. \forall \,i,j \text{ if } c_i, c_j \neq 0 \, , |\lambda_i - \lambda_j| \geq d^{-poly(m)} \, $.
\end{definition}

\noindent $\FSRelax$ decides whether the observable $\AN$ for the input initial state relaxes over a long time to the target value $A^*$. We note that the promise on the initial state strictly makes the problem easier, and so the problem without the promise is at least as hard.

The problem of determining whether the system thermalizes is closely related to determining relaxation, where the observable relaxes specifically to the expected value for a thermodynamic ensemble. We consider two relevant ensembles: the microcanonical ensemble (\cref{def:mc}) and the Gibbs ensemble.
\
The expectation value of a local observable $O$ with respect to the microcanonical ensemble is:
\begin{align}
    \langle O\rangle_{\text{MC}} \coloneqq \tr{O\rhoMC(H,E)} = \frac{1}{W} \sum_{i}  f\left(\frac{\lambda_i-E}{w}\right)\bra{\lambda_i}O\ket{\lambda_i}
\end{align}
where $W = \sum_{i}  f\left(\frac{\lambda_i-E}{w}\right)$ is a normalizing factor and $\{\ket{\lambda_i}\}$ are the eigenstates of the relevant Hamiltonian.

In the rest of the paper, we select $f(x) = e^{-\pi x^2}$, such that
\begin{equation}
\label{eq:microcanonical_gaussian}
    \rhoMC(H, E) = \frac{1}{\sum_i e^{-\pi \left(\frac{\lambda_i - E}{w} \right)^2}} \sum_i e^{-\pi \left(\frac{\lambda_i - E}{w} \right)^2} \ketbra{\lambda_i}{\lambda_i}
\end{equation}
By choosing a Gaussian window function, we ensure that the distribution of energies in the microcanonical ensemble is peaked at $E$ and decays exponentially with distance. While in principle, the microcanonical state could be defined with a rectangular window function, this has the disadvantage of making the state sensitive to the precision of $w$ in the case where there are eigenenergies close to the boundaries $\lvert E \pm w \rvert$. Previous work \cite{apeldoorn2019} has used projectors onto the eigenspace above or below a certain value to prepare the Gibbs states, but with the requirement that the value is sufficiently far from the spectrum of $H$. In order to avoid the requirement of a promise that the boundaries $\lvert E \pm w \rvert$ are away from the spectrum of $H$ and to avoid non-analyticities, we make use of a smooth window function.
For ease of notation, and since we are usually only concerned with the microcanonical ensemble at an energy corresponding to that of a single initial state, in the rest of the paper we denote $\rhoMC(E)$ as $\rhoMC$.

Another state which can represent a thermalized system is the canonical ensemble, or Gibbs state, which represents the state at a fixed temperature $(1/\beta)$.
\begin{definition}[Canonical Ensemble]
\label{def:gibbs} For a Hamiltonian $H$ the canonical ensemble, or Gibbs state $\rhog$ at inverse temperature $\beta$ is defined as
    \begin{equation}
    \rhog = \frac{e^{-\beta H}}{\tr{e^{-\beta H}}},
\end{equation}
\end{definition}

When we consider thermalization to the Gibbs state, $\beta$ is implicitly determined by the expected energy of the initial state $\ket{\psi_0}$, such that $\frac{\tr{H e^{-\beta H}}}{\tr{e^{-\beta H}}} = E$, where $E = \bra{\psi_0} H \ket{\psi_0}$. Since $\frac{\tr{H e^{-\beta H}}}{\tr{e^{-\beta H}}}$ is monotonic in $\beta$, there exists a unique $\beta$ and hence unique $\rhog$ satisfying this constraint.

We now define the decision problems for Thermalization:
\begin{definition}[(Formal) Finite Size Thermalization to the microcanonical ensemble $\FSTherm$]
\label{def:fstherm_formal}

    Consider a $k$-local Hamiltonian $H$ acting on $N$ qudits with local Hilbert space of constant dimension $d$, an observable $\AN$ which is a sum of local observables. 
    Let $c, \epsilon$ be constants.
    Let the long time average of $\AN$ be as given in \cref{Eq:Time-Averaged_Observable}.
    
    \emph{\textbf{Input:}} A classical description of product state $\ket{\psi_0}$.
    
    \emph{\textbf{Output:}} YES if $|\barAN - \tr{\AN \rhoMC(E)}|\leq \epsilon$ or NO if $|\barAN - \tr{\AN \rhoMC(E)}|\geq c\epsilon$, where $E$ is the expected energy of the input state $\bra{\psi_0} H \ket{\psi_0}$.
    
    \emph{\textbf{Promise:}} Either $|\barAN - \tr{\AN \rhoMC(E)}|\leq \epsilon$ or $|\barAN - \tr{\AN \rhoMC(E)}|\geq c\epsilon$. Let $\{\ket{\lambda_i} \}$ be the eigenstates of $H$. Then, the initial state $\ket{\psi_0} := \sum_i c_i \ket{\lambda_i}$ is only supported on energy eigenstates with an inverse exponential gap, i.e. $\exists \, m = \bigO{\mathrm{poly}(N)} \, s.t. \ \ \forall \,i\neq j \text{ if } c_i, c_j \neq 0 \, , |\lambda_i - \lambda_j| \geq d^{-poly(m)} \, $.

\end{definition}

\begin{definition}[(Formal) Finite Size Thermalization to the Gibbs ensemble $\FSThermG$]
\label{def:fstherm_gibbs_formal}

    Consider a $k$-local Hamiltonian $H$ acting on $N$ qudits with local Hilbert space of constant dimension $d$, an observable $\AN$ which is a sum of local  observables. 
    Let $c, \epsilon$ be constants.
    Let the long time average of $\AN$ be as given in \cref{Eq:Time-Averaged_Observable}.
    
    \emph{\textbf{Input:}} A classical description of product state $\ket{\psi_0}$.
    
    \emph{\textbf{Output:}}  YES if $|\barAN - \tr{\AN \rhog}|\leq \epsilon$ or NO if $|\barAN - \tr{\AN \rhog}|\geq c\epsilon$, where $\rhog$ is the Gibbs state of $H$ with expected energy equal to that of the input state $\bra{\psi_0} H \ket{\psi_0}$.
    
    \emph{\textbf{Promise:}} Either $|\barAN - \tr{\AN \rhog}|\leq \epsilon$ or $|\barAN - \tr{\AN \rhog}|\geq c\epsilon$. Let $\{\ket{\lambda_i} \}$ be the eigenstates of $H$. Then, the initial state $\ket{\psi_0} := \sum_i c_i \ket{\lambda_i}$ is only supported on energy eigenstates with an inverse exponential gap, i.e. $\exists \, m = \bigO{\mathrm{poly}(N)} \, s.t. \ \ \forall \,i\neq j \text{ if } c_i, c_j \neq 0 \, , |\lambda_i - \lambda_j| \geq d^{-poly(m)} \, $.

\end{definition}

We will also be interested in a slightly different version of the relaxation problem, where we promise that the Hamiltonian and initial state relax in polynomial amount of time in the system size.
The study of this ``quickly'' relaxing set of Hamiltonians is motivated by the fact that these are the set of Hamiltonians for which we can realistically study relaxation and equilibriation in a lab.
We call this version ``Finite-Time, Finite-Space Relaxation'', $\FTFSRelax$:
\begin{definition}[Finite-Time, Finite-Space Relaxation, $\FTFSRelax$]
    \label{def:ftfsrelax}
Consider a local Hamiltonian $H$ acting on $N$ qudits with local Hilbert space dimension $d$, a local  observable $A$, a fixed value $A^*$ and a timescale $T = \bigO{\poly{N}}$. 
Let $c, \epsilon$ be constants.
    Let the long time average of $A$ be as given in \cref{Eq:Time-Averaged_Observable}.
    
    \emph{\textbf{Input:}} A classical description of product state $\ket{\psi_0}$.
    
    \emph{\textbf{Output:}}  YES if $|\barAN - A^*|\leq \epsilon$ or NO if $|\barAN - A^*|\geq c\epsilon$.
    
    \emph{\textbf{Promise:}} Either $|\barAN - A^*|\leq \epsilon$ or $|\barAN - A^*|\geq c\epsilon$. Let $\{\ket{\lambda_i} \}$ be the eigenstates of $H$. Then, the initial state $\ket{\psi_0} := \sum_i c_i \ket{\lambda_i}$ is only supported on energy eigenstates with an inverse exponential gap, i.e. $\exists \, m = \bigO{\mathrm{poly}(N)}$ \, such that $\forall \,i,j \text{ if } c_i, c_j \neq 0 \,$ then $|\lambda_i - \lambda_j| = \Omega(1/\poly{N}) \, $.
\end{definition}

\section{\texorpdfstring{$\PSPACE$-hardness of $\FSRelax$}{PSPACE-hardness of FSRelax}}
\label{sec:hardness}
We show the $\PSPACE$-hardness of $\FSRelax$ by reducing a computationally hard problem, the finite-size halting problem $\FSHalt$, to $\FSRelax$.
\begin{definition}[Finite Size Halting ($\FSHalt$)]
    The problem $\FSHalt$ takes as input a Turing machine $M,$ an input $x$ and an integer $n$ in unary. 
    It outputs \textsc{YES} if $M$ halts on $x$ within $n$ tape space, and \textsc{NO} if it does not.
\end{definition}

By definition, $\PSPACE$ consists of all languages for which membership can be decided by a deterministic Turing machine in polynomial space. For any language $L \in \PSPACE$ decided by a Turing machine $M$ we can define a TM $M'$ that runs $M$, and halts on a Yes instance, or runs infinitely on a No instance. The language $L$ can thus be decided by $\FSHalt$ with $M'$, the instance $x$, and $n = \poly{|x|}$ as input.
Further, $\FSHalt \in \PSPACE$ since we can simply simulate $M$ with a universal Turing machine restricted to polynomial space.
Therefore, $\FSHalt$ is $\PSPACE$-complete.

We present a brief overview of our proof strategy to show that $\FSRelax$ is $\PSPACE$-hard, which is closely related to the construction used by \textcite{shiraishithermal2021}.  Our reduction relies on constructing a Hamiltonian that encodes the action of a Universal Reversible Turing Machine (URTM), and the input state to $\FSRelax$ encodes the input to the URTM. We design our Hamiltonian and observable such that the long-time average of the observable depends on whether the URTM halts or not. Therefore, the relaxation value of the observable allows us to determine the solution to instances of $\FSHalt$, and thus $\FSRelax$ must be as hard as $\FSHalt$. 
Later, in \cref{sec:FSThermHardness}, we show how to modify this construction to apply to $\FSTherm$.
We discuss the relationship between our hardness construction and that of \textcite{shiraishithermal2021} in \cref{sec:discussion}.

\subsection{Turing Machines and Setup}
\label{Sec:Turing_Machine_Description}

We construct a Turing machine $\mathcal{M}$ composed of three TMs: 
\begin{itemize}
    \item TM1, a 10-state, 8-symbol URTM  \cite{Morita2017}. The states of the finite control are $Q_u$ and the symbols for the tape are $\Gamma_u$ 
    \item TM2, which has 1 state $r$ and 2 symbols $(a_1, a_2)$ 
    \item TM3, which has 10-states $Q_{rev}$, and uses the same 8-symbols $\Gamma_u$ to run TM1 in reverse.
\end{itemize}
Each site of our physical system represents one cell or the finite control of the TM tape. Our system is chosen to have periodic boundary conditions, so our tape is also periodic. Similar to \textcite{shiraishithermal2021}, we denote the cells with symbols from TM1 and TM3 as $M$-cells, and the cells with symbols of TM2 as $A$-cells. 

Our TM $\mathcal{M}$ functions as follows.
Assume in the initial state we fix a constant $\alpha$ as the fraction of cells which are to be $M$-cells, and we evenly distribute $M$ and $A$ cells. 
The initial state has the input to TM1 directly on the $M$ cells. Every $A$ cell starts with the symbol $a_1$. TM1 commences on the M cells, and when it encounters an $A$ cell it simply skips over in the same direction. TM1 functions as a URTM on the M cells. 
If TM1 halts, the finite control state changes to the state $r$ and TM2 commences, flipping every A cell from $a_1$ to $a_2$. Eventually, TM2 wraps around the tape and encounters an $A$ cell with the symbol $a_2$ -- this is the first cell encountered by TM2, adjacent to the halting cell of TM1. Then, the finite control state is changed to a state of TM3, which starts on the halting cell of TM1 and runs TM1 in reverse, before finally halting. By running TM1 in reverse, but with the A cells flipped, we ensure that in the case of halting, the A cells remain flipped for a sufficiently long duration.
TM3 has a special halting state which it enters at the point when it would transition to the initial state of TM1.

In order to ensure that our Hamiltonian is 2-local, we split each step of TM1 and TM3 into 2 - first to change the control state, and then to move the control either left or right. 
Thus the set of symbols and states required is:
\begin{align}
    Q &= \{W,M\}\times (Q_u \cup Q_{rev}) \cup  \{r\},\nonumber \\
    \Gamma &= \Gamma_u \cup \{a_1,a_2\},
\end{align}
the ``$W$" corresponding to the write substep and the ``$M$'' corresponding to the move substep.

We note $|Q_u|=10, |\Gamma_u|=8$, and $|Q_{rev}|=10$.
To ensure the combined TM functions as intended, we furnish the combined TM with an additional transition rule which takes it out of the halting state of TM1, and moves its control to the state $r$, thereby initiating TM2, as well as a transition rule for when TM2 encounters $a_2$, which then halts TM2 and initiates TM3.
Thus, if TM1 does not halt on its input, $\mathcal{M}$ has all A-cells filled with $a_1$, and if TM1 halts on its input, $\mathcal{M}$ has all A-cells filled with $a_2$.

\subsection{TM-to-Hamiltonian Mapping} \label{Sec:TM-to_Hamiltonian}

We now need find a way of mapping our Turing Machine $\mathcal{M}$ to a Hamiltonian.
Fortunately, this is a well studied problem, originally in Ref. \cite{Feynman_1986}, and a variety of TM-to-Hamiltonian mappings have been developed for proving hardness of properties of Hamiltonians \cite{breuckmann2014space,bausch2017complexity, watson2019detailed}.

Let $U$ be the unitary which implements the transition rule of a given TM.
In particular, if there is a transition rule $\ket{ab}\rightarrow \ket{cd}$, then we define the term $U_{abcd} =\ket{cd}\bra{ab}$, applied to all relevant qudits.
We then define the set of transition rules to be $\mathcal{T}$ and the corresponding Hamiltonian to be:
\begin{align}
    H = \sum_{abcd\in \mathcal{T}} (U^\dagger_{abcd} +U_{abcd}).
\end{align}
Now let $\ket{w_0}$ be an initial state of some computation in the computational basis.
We can then define the state after the $k^{th}$ application of the TM transition rule as $\ket{w_k} = V^k\ket{w_0},$ where $V:=\sum_{abcd\in\mathcal{T}}U_{abcd}$ maps product states representing some Turing machine configuration to those of the Turing machine's incrementation by the reversibility of the transition relations described by $\mathcal{T}$.

Given a particular input state $\ket{w_0}$, the encoded computation either (a) halts, (b) loops forever.
We can apply a unitary transformation to this get the Hamiltonian in a block-diagonal form, where each of the blocks corresponds to a different input state (see \cite{kitaev2002classical} or \cite[Section 5]{watson2019detailed}).
Within each block, the Hamiltonian has the structure of a graph Laplacian of either a line graph (halting case) or loop (non-halting case).
Thus for a given input state, we can define an effective Hamiltonian $\Heff$ corresponding to the particular input state:
\begin{align}\label{Eq:Effective_Hamiltonian}
    \Heff = \sum_{k=0}^{T-2}\ket{w_k+1}\bra{w_k}+ h.c.,
\end{align}
where $T$ is either the halting time, or the time until the TM begins to loop.
This effective Hamiltonian has eigenstates of
the adjacency matrix of a line graph or loop which are \cite{yueh2008explicit,da2020eigenpairs}: 
\begin{align}\label{Eq:Energies}
    \lambda_j(\Heff) = 2\cos\left(\frac{(j+1)\pi}{T+1} \right)
\end{align} for for halting and 
\begin{align}
    \lambda_j(\Heff) &= 2\cos\left(\frac{2\pi j}{T}\right)
\end{align}
for non-halting, $j\in\{0,1,\dots, T-1\}.$
Due to the determinism of Turing Machine evolution and the finite number of configurations of the machine, we see that $T= \bigO{d^N}.$ 

\subsection{Long-Time Expectation Values} \label{Sec:Long-Time_Expectation_Values}

Given the TM-to-Hamiltonian mapping presented in \cref{Sec:TM-to_Hamiltonian}, we can encode the TMs given in \cref{Sec:Turing_Machine_Description}.
We now consider an initial state:
\begin{align}
    \ket{w_0} = \ket{q_0}\otimes \ket{y},
\end{align}
where $q_0$ corresponds to the the TM head in its initial state and $y\in \Gamma^N,$ with a proportion $(1-\alpha)$ evenly spaced out A-cells in the state $\ket{a_1}$. On the sites that are not set to $a_1,$ $y$ will encode Turing machine we wish to determine $\FSHalt$ for into the URTM's tape symbols. We then choose the TM1 to run the problem $\FSHalt$   .
If the TM1 halts, we have it transition to TM2 which flips $\ket{a_1}\rightarrow \ket{a_2}$.  For the remainder of \cref{Sec:Turing_Machine_Description}, we will define $A_i:=\ket{a_2}\bra{a_2}_i.$

\begin{lemma}
	\label{lem:halting_separation}
    If $y$ encodes a Turing Machine that either halts in time $T_h$ or never halts, we have that
    \begin{align}
        \barAN&=\begin{cases}
        \frac{1-\alpha}{2}\left[1-\frac{2}{2T_h+N+1}\right] & \text{halting}\\
        0 & \text{non-halting}.
    \end{cases}
    \end{align}
\end{lemma}
\begin{proof}
We begin by considering the limit definition for  $\barAN:$
\begin{align}
    \barAN &= \lim_{T\to\infty}\frac{1}{T}\int_{0}^Tdt \bra{w_0}e^{iHt}\AN e^{-iHt}\ket{w_0}\nonumber\\
    &= \lim_{T\to\infty}\sum_{ij}\frac{1}{T}\int_{0}^Tdt e^{i(\lambda_i-\lambda_j)t}c_i^*c_j\bra{\lambda_i}\AN\ket{\lambda_j}\nonumber\\&=\sum_{ij}\delta_{\lambda_i,\lambda_j}c_i^*c_j\bra{\lambda_i}\AN\ket{\lambda_j} \nonumber\\&= \sum_{ij}\delta_{i,j}c_i^*c_j\bra{\lambda_i}\AN\ket{\lambda_j} \nonumber\\&=\sum_{j}|c_j|^2\bra{\lambda_j}\AN\ket{\lambda_j}, \label{Eq:barAN_Eigenstates}
\end{align}
where for the penultimate equality we have used that none of the eigenvalues of the effective Hamiltonian have the same energy (as per \cref{Eq:Energies}).

 If TM1 never halts, then TM2 never starts and none of then $a_1$ cells are ever flipped to $a_2$ and we see that
\begin{align}
    \AN\ket{w_j} = 0, \quad \forall j.
\end{align} 
This in turn means $\AN\ket{\lambda_j}=0$ for all $j,$ as all $\ket{\lambda_j}$ are superpositions of the $\ket{w_k}.$ Thus, for non-halting, $\barAN=0.$
\textcite[Eq. S.46]{shiraishithermal2021} show that for halting, $\barAN$ becomes:
\begin{align}
    \sum_{j}|c_j|^2\bra{\lambda_j}\AN\ket{\lambda_j} &= \frac{3}{2(T+1)}\left( \bra{w_0}\AN\ket{w_0}+\bra{w_{T-1}}\AN\ket{w_{T-1}}    \right)  \nonumber\\
    &+ \frac{1}{T+1}\sum_{j=1}^{T-2}\bra{w_j}\AN\ket{w_j} - \frac{1}{2(T+1)}\sum_{\substack{0\leq j,j'\leq T-1 \\ j'=j\pm2}} \bra{w_j}\AN\ket{w_{j'}} \label{Eq:Expectation_Value}. 
\end{align} 
We now analyze the terms of this equation. 
We have that the $M$ cells are evenly spaced. We assume that every $M$ cell is followed by $\alpha^{-1}-1$ $A$-cells, and that $\alpha N$ of these $\alpha^{-1}$ blocks constitute the whole Turing Machine. Thus, noting that $\bra{w_0}\mathcal{A}_N\ket{w_0}=0,$ $\bra{w_{T-1}}\mathcal{A}_N\ket{w_{T-1}}=1-\alpha,$ and the final term in \cref{Eq:Expectation_Value} is $0,$ it remains to calculate 
\begin{align}
    \sum_{j=1}^{T-2} \bra{w_j}\mathcal{A}_N\ket{w_j} &= (1-\alpha)\left[T-2-(T_h+N+1)+1\right]+\sum_{j=T_h+1}^{T_h+N}\bra{w_j}\mathcal{A}_N\ket{w_j}\nonumber\\
    & \label{eq:expectation_value2} = (1-\alpha)\left[T-2-(T_h+N)\right] + \sum_{j=T_h+1}^{T_h+N}\bra{w_j}\mathcal{A}_N\ket{w_j}.
\end{align}
If $\ket{w_j}$ has $l$ $A$-cells with the symbol $a_2$, then 
\begin{equation}
	\bra{w_j}\mathcal{A}_N\ket{w_j} = l\bra{a_2} A \ket{a_2}.
\end{equation}
Due to the uniformity of the placement of the $M$-cells (and recalling our assumption that $\alpha N\in\mathbb{Z}$), for $j \in [T_h+1, T_h+N]$,
\begin{align}
\sum_{j=T_h+1}^{T_h+N}\bra{w_j}\mathcal{A}_N\ket{w_j} &= \sum_{j=1}^{N}\frac{\left[j-\lfloor \alpha j\rfloor \right]}{N}\bra{a_2} A \ket{a_2}\nonumber\\
&= \frac{1}{N}\left[\sum_{j=1}^Nj-\sum_{j=1}^N\lfloor \alpha j\rfloor\right]\bra{a_2} A \ket{a_2}\nonumber\\
&= \frac{1}{N}\left[\frac{N(N+1)}{2} - \sum_{j=0}^{N-1}\lfloor \alpha j\rfloor - N\right]\bra{a_2} A \ket{a_2}\nonumber\\
&=\frac{1}{N}\left[\frac{N(N+1)}{2}-\alpha^{-1} \sum_{k=0}^{\alpha N-1}k-N\right]\bra{a_2} A \ket{a_2}\nonumber\\
&= \left[\frac{N+1}{2} -\frac{1}{\alpha N} \frac{\alpha N(\alpha N-1)}{2}-1\bra{a_2} \right]A \ket{a_2}\nonumber\\
&= \frac{(1-\alpha)N}{2}\bra{a_2} A \ket{a_2}.
\end{align}
Thus,
\begin{align}
    \sum_{j=1}^{T-2}\bra{w_j}\mathcal{A}_N\ket{w_j} &= (1-\alpha)[T-(T_h+N/2+2)]\bra{a_2} A \ket{a_2}.
\end{align}
Combining this with \cref{Eq:Expectation_Value} and \cref{eq:expectation_value2}, we obtain
\begin{align}
    \barAN &= (1-\alpha)\frac{T-(T_h+N/2+1/2)}{T+1} \bra{a_2} A \ket{a_2} \label{eq:expvaluegeneral}.\end{align}
Finally, noting that $T=2T_h+N,$ and $A = \ketbra{a_2}{a_2}$ we obtain 
\begin{align}
    \barAN &= (1-\alpha)\frac{T_h+N/2-1/2}{2T_h+N+1}\nonumber\\
    &= \frac{1-\alpha}{2}\left[1-\frac{2}{2T_h+N+1}\right].
\end{align}
\end{proof}

\subsection{Putting it All Together and Verifying the Promise }

We now wish to use our previously proved results about the Hamiltonian to prove hardness of $\FSRelax$. 
For the remainder of this section, we remember that $A$ satisfies $\bra{a_2}A\ket{a_2}=1$.

\begin{theorem}
\label{thm:fsrelax_pspace_hardness}
    $\FSRelax$ \  is PSPACE-hard for $d=\LocalDim$.
\end{theorem}
\begin{proof}
    We consider the Hamiltonian generated from the TM-to-Hamiltonian mapping, where the TM is chosen to be that from \cref{Sec:Turing_Machine_Description}.
    We then choose an initial state $\ket{\psi_0} = \ket{q_0}\otimes \ket{y}$, where $y$ is a string $y\in \Gamma^{\times N-1}$.
    In particular, we choose $y$ to have a fraction $\alpha$ M cells, where we are free to choose $\alpha$.
    We choose the $M$ first $n$ cells to be a bit string of length $n$, $x\in \{0,1\}^{n}$.
    
    We then choose the first TM encoded to run the $\FSHalt$ problem.
    The action of the final TM only occur if the initial TM halts. 
    We see from \cref{lem:halting_separation} that the time-averaged expectation value is:
    \begin{align}
   	\barAN&=\begin{cases}
   		\frac{1-\alpha}{2}\left[1-\frac{2}{2T_h+N+1}\right] & \text{halting}\\
   		0 & \text{non-halting}.
   	\end{cases}
   \end{align}
    By choosing $\alpha = \Theta(1)$ to be a sufficiently small constant we get, for sufficiently large $N$, separated values of $\barAN$ in the halting and non-halting cases respectively.
    Since $\FSHalt$ \  is $\PSPACE$-hard, then determining which case will occur for $\barAN$ is also $\PSPACE$-hard.
    We show later in \cref{Lemma:Promise_Satisfied} that the promise that $|\lambda_i-\lambda_j|>d^{-\poly{m}}$ is satisfied by this Hamiltonian and input state.
\end{proof}

 To find the local Hilbert space dimension of the Hamiltonian, we simply need to add up the necessary Hilbert space size for the TMs to function.
 We realise that the total Hilbert space size will be the size of all the alphabet symbols plus the size of all the TM head state symbols.
 Thus we expect the local Hilbert space dimension to be:
 \begin{align}
     d &= 2\times( |Q_u|+|Q_{rev}|) + |\{r\}| + |\Gamma_u|+ |\{a_1,a_2\}| \nonumber\\
     &= \LocalDim.
 \end{align}

\noindent Finally we verify that the Hamiltonian and initial state we have used for our hardness construction satisfies the promise of $\FSRelax$.
 \begin{lemma}[Promise is Satisfied]\label{Lemma:Promise_Satisfied}
    Let $\ket{w_0}$ be an initial basis state, and let $H$ be a Hamiltonian from a TM-to-Hamiltonian mapping on $N$ qudits.
    Then:
    \begin{align}
        \min_{ij}|\lambda_i(\Heff) - \lambda_j(\Heff)| = \Omega\left( d^{-2N} \right).
    \end{align}
\end{lemma}
\begin{proof}
    From \cref{Eq:Energies}, we see that the minimum difference in energies occurs when for $j=i+1$, giving:
    \begin{align}
        |\lambda_{j+1}(\Heff) - \lambda_j(\Heff)| &=  2\left(\cos\left(\frac{(j+1)\pi}{T+1}\right)-\cos\left(\frac{j\pi}{T+1}\right)\right)\nonumber\\
        &= 4\left\vert \sin\frac{\pi}{2(T+1)}\sin\frac{(j+1)\pi}{2(T+1)}\right\vert\nonumber\\
        &\geq \frac{\pi^2}{4(T+1)^2}\nonumber\\
        &= \Omega(T^{-2}) = \Omega(d^{-2N}).
    \end{align}
\end{proof}

\section{Containment of \texorpdfstring{$\FSRelax$ in $\PSPACE$}{FSRelax in PSPACE}}
\label{sec:containment}
We show that $\FSRelax \in \PSPACE$ by demonstrating a $\poly{N}$-space quantum algorithm to estimate $\barAN$. Combined with the result that $\BQPSPACE \subseteq \PSPACE$  \cite{Watrous_1999}, and using the fact that $\FSTherm$ is an instance of $\FSRelax$, this shows that $\FSTherm \in \PSPACE$.
\begin{lemma}\label{lemma:sqrtblock}
    Let $\AN$ be defined as in \cref{eq:ANdef}. 
    For any $\epsilon \in [0,1/4],$ there exists a quantum circuit to prepare a $(4,N\lceil \mathrm{log}_2d\rceil +5,\epsilon)$-block encoding of $\sqrt{1-\AN/2dN}$
    using $\bigO*{\poly{\log{N}}\poly{\log{\epsilon^{-1}}}}$ 
    ancillas, with a classical description that can be computed in time $\bigO{\poly{N}\poly{\log{\epsilon^{-1}}}}.$
\end{lemma}
\begin{proof}
    We begin by noting that $\AN$ is a $dN$-sparse matrix, and that all its entries are less than or equal to $1$ in magnitude. Thus, by \cref{lemma:sparseblockenc}, we have that we can produce a $(dN, \lceil \mathrm{log}_2(d)\rceil N+3, \epsilon_2)$-block encoding of $\AN$ with $\bigO*{\poly*{\mathrm{log}\left(\frac{d^2N^2}{\epsilon_2}\right)}}$ ancillas. 
    
    By \cref{lemma:polyapprox_sqrt}, for all $\epsilon_1\leq 1/4,$ there exists a polynomial $p(x)$ of degree $\bigO{\log{\epsilon_1^{-1}}}$ such that 
    \begin{equation}
    \sup_{x\in [-1,1]}\left\vert\frac{1}{4}\sqrt{1-x/2}-p(x)\right\vert\leq \epsilon_1,
    \end{equation}
    and $\sup_{x\in [-1, 1]}\vert p(x)\vert \leq 1/2.$ Thus, by \cref{lem:approx_poly_block_encoding}, we can prepare a $\left(1, \lceil \mathrm{log}_2(d)\rceil N\!+\!5, \bigO*{\log{\epsilon_1^{-1}}\sqrt{\epsilon_2/dN}\!+\!\delta}\right)$-block encoding of $p\left(\frac{\AN}{dN}\right)$ with $\bigO*{\log{\epsilon_1^{-1}}\poly*{\mathrm{log}\left(\frac{d^2N^2}{\epsilon_2}\right)}}$ ancillas, $\bigO{\log{\epsilon_1^{-1}}(\poly{N}+\poly{\log{\epsilon_2^{-1}}})}$ fundamental gates, and a classical description that can be computed in time 
    \begin{equation}
    \bigO*{\log{\epsilon_1^{-1}}\left(\poly{N}+\poly*{\log{\epsilon_2^{-1}}}\right)+\poly*{\log{\epsilon_1^{-1}}}+\poly{\log{\delta^{-1}}}}.
    \end{equation}

    Thus, to obtain the desired result, it suffices to take $\delta, \epsilon_1, \log{\epsilon_1^{-1}}\sqrt{\epsilon_2/dN}=\bigtheta{\epsilon}.$ To do this, we see that we can take $\delta,\epsilon_1=\bigtheta{\epsilon},$ and $\epsilon_2 = \bigtheta*{\left(\frac{\sqrt{N}\epsilon}{\log{\epsilon^{-1}}}\right)^2}.$
    
\end{proof}
We use this block encoding circuit to estimate $\barAN$.
\begin{lemma}
\label{lem:long_time_qpoly}
For any $0<\epsilon\leq 1/4,$ there is a quantum algorithm that takes space $\bigO*{\epsilon^{-2}\poly{\log{\epsilon^{-1}}}\poly{N}}$ 
to return $\Gamma$ such that
\begin{equation}
    |\Gamma - \barAN| \leq \epsilon
\end{equation} with probability of at least $2/3.$
\begin{proof}

Note that we can uniquely write our initial state as $\ket{\psi_0}=\sum_{i}c_i\ket{\lambda_i}$ with $\ket{\lambda_i}$ a $\lambda_i$-eigenvector of $H$ such that $\lambda_i\neq \lambda_j$ for all $i\neq j$. Further, by the promise condition, we have that there exists 
$m=\poly{N}$ that satisfies $\min\left\{d^{-N},\inf_{i\neq j}\vert \lambda_i-\lambda_j\vert\right\} > 2^{-m}$ for all $i$ for which $c_i\neq 0.$
We begin by noting that we can implement a unitary $\tilde{U}_{\epsilon_1}$ such that $\Vert \tilde{U}_{\epsilon_1}-e^{iH}\Vert \leq \epsilon_1 2^{-5(m+1)}$ in polynomial space with an efficient classical description via standard Trotterization techniques \cite{Childs_2021}.
If we perform phase estimation on $\tilde{U}_{\epsilon_1}$ with $\mathcal{O}(m)$ bits of phase precision, we would then get an approximation to the phase estimation circuit on $e^{iH}$ with (operator norm) error at most $\epsilon_1.$

Let $\ket{\tilde{\lambda}_i^{(m,\epsilon_22^{-(1+m)})}}$ be the output in the phase register of running quantum phase estimation of $e^{iH}$ on $\ket{\lambda_i},$ desired to be accurate to $m$-bits of precision with error parameter $\epsilon_2 2^{-(m+1)}.$ By \cref{lemma:phaseestspacecomp},  we have that this can be achieved by a circuit with $3(m+1)+ \lceil 2\mathrm{log}_2 \epsilon_2^{-1}\rceil$ total qubits. Despite the required exponential depth of the phase estimation circuit, it has an efficient classical description due to its repeated structure - so a description of each layer of gates can be computed in polynomial time classically.   
Then, by \cref{cor:QPEOrtho} we have:
\begin{equation}
    \left\vert\braket{\tilde{\lambda}_i^{(m,\epsilon_2 2^{-(m+1)})}}{\tilde{\lambda}_j^{(m,\epsilon_2 2^{-(m+1)})}}- \delta_{i,j}\right\vert\leq  \epsilon_2 2^{-m},
\end{equation}
and hence,
\begin{equation}
    \left\vert \sum_{i,j=0}^{\vert\sigma(H)\vert-1} c_i c_j^* \bra{\lambda_j} \AN\ket{\lambda_i}\braket{\tilde{\lambda}_i^{(m,\epsilon_2 2^{-(m+1)})}}{\tilde{\lambda}_j^{(m,\epsilon_2 2^{-(m+1)})}} -\sum_{i=0}^{\vert\sigma(H)\vert-1} \vert c_i\vert^2 \bra{\lambda_i} \AN
    \ket{\lambda_i} \right\vert 
    \leq
    \epsilon_2. 
\end{equation} 
The long time average of $\AN$ for our initial state $\ket{\psi_0}$
can be  expressed as: 
\begin{align}
    \barAN &= \lim_{T\rightarrow \infty}\frac{1}{T}\int_0^T dt \sum_{ij}e^{-i(\lambda_i-\lambda_j)t}c_i^*c_j\bra{\lambda_i}\AN\ket{\lambda_j} \nonumber\\
    &= \sum_{i=0}^{\vert \sigma(H)\vert-1} \vert c_i\vert^2 \bra{\lambda_i}\AN\ket{\lambda_i}.
\end{align}
    We now describe a quantum $\bigO{\poly{N}}$ space protocol to estimate $\barAN$. An outline of our algorithm is shown below:

\begin{align}
        \ket{0^a}
\left(\ket{\psi_0}\ket{0^{3(m+1)+\lceil 2\mathrm{log}_2\epsilon_2^{-1}\rceil}}
    \right)
    =&
    \ket{0^a}
    \left(\sum c_i \ket{\lambda_i}\ket{0^{3(m+1)+\lceil 2\mathrm{log}_2\epsilon_2^{-1}\rceil}}
    \right) 
    \nonumber\\
    \xrightarrow{\text{Phase estimation of } \tilde{U}_{\epsilon_2}}
    &\ket{0^a}
    \left(\sum c_i \ket{\lambda_i}\ket{\tilde{\lambda}_{i}^{(m,\epsilon_22^{-(m+1)})}}
    +\ket{\zeta}\right)
    \nonumber\\
    \xrightarrow{\text{Block encoding $\frac{1}{4}\sqrt{1-\AN/2dN}$}}
    &\ket{0^a}
    \left(\sum  \frac{c_i}{4}\sqrt{1-\AN/2dN}\ket{\lambda_i}\ket{\tilde{\lambda}_{i}^{(m,\epsilon_22^{-(m+1)})}}
    + \ket{\zeta'}\right) + \ket{\xi} \nonumber\\
    \xrightarrow{\text{Measure first register}}
     &\mathrm{Pr}[0^a]\approx \frac{1}{16}\sum_i \vert c_i\vert^2\bra{\lambda_i}(1-\AN/2dN)\ket{\lambda_i} = \frac{1-\barAN/2dN}{16}. 
\end{align}

First, we prepare the state 
$\ket{0^a}
\otimes \left(\ket{\psi_0}\ket{0^{3(m+1)+\lceil 2\mathrm{log}_2\epsilon_2^{-1}\rceil}}\right)
$. 
The first register has $a = \lceil N\mathrm{log}_2d\rceil+5=\bigO{N}$ qubits. 
Since $m = \bigO{\poly{N}}$, this state can be prepared in $\bigO{\poly{N}}$ space. For convenience, we use a qubit-based algorithm rather than a qudit based algorithm, where each qudit of the original state is replaced with $\lceil{\log d} \rceil$ qubits. We designate the first $a$ qubits, which will act as a control, as register 1, and the remaining register as register $2.$ 

Next, we perform phase estimation on register 2
for the operator $\tilde{U}_{\epsilon_2}$
, taking it to $m$ bits of precision and with $2$-norm error in the approximation register of $\epsilon_22^{-(m+1)}$.
The resulting state at the end of this step is:
\begin{equation}
    \ket{0^a}\left(\sum c_i \ket{\lambda_i}\ket{\tilde{\lambda}_i^{(m,\epsilon_22^{-(m+1)})}}
    +\ket\zeta\right)
    ,
\end{equation}
where $\ket{\tilde{\lambda}_i^{(m,\epsilon_22^{-(m+1)})}}$ is as above, and $\ket\zeta,$ the error from performing the QPE on $\tilde{U}_{\epsilon_1}$ instead of $e^{iH},$ satisfies $\Vert\ket\zeta\Vert\leq \epsilon_1.$ We are once again able to give an efficient classical description of this phase estimation. 

We now apply a block encoding of $\sqrt{1-\AN/2dN}
/4$ to our state. By \cref{lemma:sqrtblock}, we can do so to $\epsilon_3$ error using $\bigO{\poly{\log{N}}\poly{\log{\epsilon_3^{-1}}}}$ ancillas. 
The resulting state is:
\begin{equation}
   \ket{0^a}
   \left(\sum \frac{c_i}{4} \sqrt{1-\AN/2dN}
   \ket{\lambda_i}\ket{\tilde{\lambda}_{i}^{(m,\epsilon_22^{-(m+1)})}}
   +\ket{\zeta'}\right) 
    + \ket{\xi}
\end{equation}
Where $(\bra{0^a}
\otimes I)\ket{\xi} =0,$ and $\Vert \ket{\zeta'}\Vert \leq \epsilon_1+\epsilon_3.$

Thus, we see that if we measure the first register, we will obtain $0^a$ with probability 
\begin{align}
    \mathrm{Pr}[0^a] &= \sum_{i,j}\frac{c_i^*c_j}{16} \bra{\lambda_i}\left(1-\AN/2dN\right)\ket{\lambda_j}\braket{\tilde{\lambda}_i^{(m,\epsilon_22^{-(m+1)})}}{\tilde{\lambda}_j^{(m,\epsilon_22^{-(m+1)})}}\nonumber\\&\qquad + \sum_{i} \left[\frac{c_i}{4} \left(\bra{\zeta'}\sqrt{1-\AN/2dN}\otimes I\right)\ket{\lambda_i}\ket{\tilde{\lambda}_i^{(m,\epsilon_22^{-(m+1)}}}\!+\!\frac{c_i^*}{4} \bra{\lambda_i}\bra{\tilde{\lambda}_i^{(m,\epsilon_22^{-(m+1)}}}\left(\sqrt{\AN}\!\otimes \!I\ket{\zeta'}\right)\!\right]\!,
\end{align}
and hence,
\begin{equation}
    \left\vert \mathrm{Pr}[0^a]-\frac{1-\frac{\barAN}{2dN}}{16}\right\vert\leq \frac{\epsilon_2}{16}+\frac{1}{2}(\epsilon_1+\epsilon_3).
\end{equation}

Thus, by Hoeffding's inequality, for any $\epsilon_4>0,$ it suffices to repeat this measurement $\lceil \epsilon_4^{-2}\rceil$ times and take the sample average to get an $(\epsilon_2/16+(\epsilon_1+\epsilon_3)/2+\epsilon_4)$-estimate of $\frac{1-\frac{\barAN}{2dN}}{16}$ with probability at least $2/3.$ To get overall error of $\epsilon$ in the estimate of $\vert \barAN\vert$ with that probability, it suffices to take $\epsilon_2=\epsilon/8dN$, $\epsilon_4=\epsilon/128dN,$ and $\epsilon_1,\epsilon_3=\epsilon/64dN.$ Thus, we see that to repeat this measurement, we need in total \begin{align}\bigO*{\epsilon^{-2}\left(\poly{N}+\poly{\log{N\epsilon^{-1}}}+\poly{\log{N}}\poly{\log{N\epsilon^{-1}}}\right)} &=\bigO{\epsilon^{-2}\poly{\log{\epsilon^{-1}}}\poly{N}}\end{align} space. 
\end{proof}    
\end{lemma}
Finally, this algorithm allows us to decide $\FSRelax$. Hence, we have
\begin{lemma}
\label{lem:fsrelax_pspace_containment}
$\FSRelax \in \PSPACE$.
\begin{proof}
\cref{lem:long_time_qpoly} allows us to construct a probabilistic quantum polynomial space algorithm to decide $\FSRelax$. The algorithm in \cref{lem:long_time_qpoly} relies on postselection.
    We can use \cref{lem:long_time_qpoly} to probabilistically compute $\Gamma$ such that $|\Gamma - \barAN| \leq \epsilon'$ for some $\epsilon' < \epsilon$. In the event of the correct postselection outcomes occurring, since we have the promise that either $|\barAN - A^*| \leq \epsilon$ or $|\barAN - A^*| \geq c\epsilon$, one of the two following cases holds.
    \begin{enumerate}
        \item If $|\Gamma - A^*| \leq \epsilon + \epsilon'$ then $|\barAN - A^*| \leq \epsilon.$
        \item Otherwise,if $|\Gamma - A^*| \geq c \epsilon - \epsilon'$, then $|\barAN - A^*| \geq c\epsilon.$
    \end{enumerate}
    Hence if our algorithm postselects correctly, $\FSRelax$ can be decided by a polynomial space quantum algorithm that computes $|\Gamma - A^*|$.
    On the other hand, if our algorithm does not postselect correctly, we can simply return YES or NO at random. For YES instances of $\FSRelax$, we accept with probability strictly greater than $1/2$, while for NO instances, we accept with probability strictly less than $1/2$. Hence, $\FSRelax \in \PrQSPACE$. 
    By \textcite{Watrous_1999}, $\PrQSPACE \subseteq \PSPACE$. Therefore, $\FSRelax \in PSPACE$.
\end{proof}
\end{lemma}
\noindent Finally, our main result \cref{thm:fsrelax_pspace_complete} follows immediately from \cref{lem:fsrelax_pspace_containment} and \cref{thm:fsrelax_pspace_hardness}.

\section{\texorpdfstring{$\FSTherm$ and $\FSThermG$ are contained in $\PSPACE$}{FSTherm(MC) and FSTherm(Gibbs) are contained in PSPACE}}
\label{sec:fstherm_containment}
In this section we'll show that $\FSTherm$ and $\FSThermG$  are contained in $\PSPACE$. We will first show a polynomial quantum space algorithm to compute $\tr{\AN \rhoMC}$. Along with the polynomial space algorithm to determine $\FSRelax$, this yields a polynomial quantum space algorithm to decide $\FSTherm$ by reducing any instance of $\FSTherm$ to an instance of $\FSRelax$. Similarly, using a polynomial space quantum algorithm to compute $\tr{\AN \rhog}$, we show that $\FSThermG \in \PSPACE$.  Throughout this section, we restrict to systems of qubits rather than qudits, in order to make use of qubit block-encoding results. However, this containment result is general, since a system of a $k$-local Hamiltonian on $N$ qudits of local dimension $d$ can be described as a system of a $k\log{\lceil d \rceil}$-local Hamiltonian on $N\log{\lceil d \rceil}$ qubits.
 
First, we prove some preliminary lemmas about block-encodings, which we will use to prepare an approximation of the microcanonical state.
 \begin{proposition}
 \label{prop:poly_blockencoding}
     Let $H$ be a Hermitian matrix on $N$ qubits, and $\alpha \geq \norm{H}$. Let the eigenvectors and eigenvalues of $H$ be $\{\ket{\lambda_i}, \lambda_i\}$. For a polynomial $p$, let $U$ be an $(N+m)$-qubit, $(\alpha, m, 0)$-block encoding of $p(H/\alpha)$. 
     Let $\ket{\psi} \coloneqq \frac{1}{\sqrt{2^N}} \sum_i  \ket{\lambda_i}_a\ket{\lambda_i}_b$, where $N$ qubits form the main computational register $a$, and $N$ qubits form an auxiliary register $b$.
     Then, $U$ can be used to prepare the state
     \begin{equation}
         \ket{\psip(p)} \coloneqq  \frac{p(H/\alpha)_a\ket{\psi}}{\norm{p(H/\alpha)_a\ket{\psi}}}.
     \end{equation}
     such that
     \begin{equation}
         \rhop(p) \coloneqq  \ptr{b}{\ketbra{\psip(p)}{\psip(p)}} = \sum_i \frac{(p(\lambda_i/\alpha))^2}{\sum_j (p(\lambda_i/\alpha))^2} \ketbra{\lambda_i}{\lambda_i}
     \end{equation}
 \end{proposition}
 \begin{proof}
     Let $c$ be the register of the $m$ block encoding ancilla qubits. $U$ acts non-trivially on registers $a$ and $c$.
    By \cref{def:block_encoding}, $U$ prepares the state
    \begin{equation}
        U\ket{\psi} = p(H/\alpha)_a \ket{\psi}\ket{0^m}_c + \ket{\phi^{\perp}}, 
    \end{equation}
    where $\bra{0^m}_c\ket{\phi^{\perp}} = 0$, and 
    \begin{align}
        p(H/\alpha)_a \ket{\psi} &= \frac{1}{\sqrt{2^n}} \sum_i \frac{1}{\sqrt{2^N}} p(H/\alpha)_a \ket{\lambda_i}_a\ket{\lambda_i}_b \nonumber\\
        &= \frac{1}{\sqrt{2^N}} \sum_i  p(\lambda_i/\alpha) \ket{\lambda_i}_a\ket{\lambda_i}_b.
    \end{align}
    The state obtained after postselecting on the correct block-encoding being performed is
    \begin{align}
        \ket{\psip(p)} &\coloneqq  \frac{p(H'/\alpha)_a\ket{\psi}}{\norm{p(H'/\alpha)_a\ket{\psi}}} \nonumber\\
        &= \frac{1}{\sqrt{\sum_j (p(\lambda_i/\alpha))^2}} \sum_i p(\lambda_i/\alpha) \ket{\lambda_i}_a\ket{\lambda_i}_b.
    \end{align}
    It follows that
    \begin{equation}
         \rho_p \coloneqq  \ptr{b}{\ketbra{\psip(p)}{\psip(p)}} = \sum_i \frac{(p(\lambda_i/\alpha))^2}{\sum_j (p(\lambda_i/\alpha))^2} \ketbra{\lambda_i}{\lambda_i}.
     \end{equation}
 \end{proof}
 We now show that there exists a suitable choice of polynomial $p$ such that $\rho_p$ approximates the microcanonical state.
\begin{lemma}
    \label{lem:polyapprox_microcanonical}
    Let $H$ be a Hermitian matrix on $N$ qubits, and $\alpha \geq \norm{H}$. Let $\pgauss{w/\alpha}{\eta}$ be defined as in \cref{lemma:polyapprox_gaussian}. Then, for any $\epsilon_p \in (0, 1]$, with $\eta = \frac{\epsilon_p e^{-\pi \alpha^2/w^2}}{9}$, $\pgauss{w/\alpha}{\eta}$ has degree $\bigO{\frac{\alpha}{w}(\log{1/\epsilon} + \pi \alpha^2/w^2)}$ and
    \begin{equation}
        \norm{\rhop(\pgauss{w/\alpha}{\eta}) - \rhoMC(H,0)}_1 \leq  \epsilon_p.
    \end{equation}
\end{lemma}
\begin{proof}
    First, we note that by \cref{lemma:polyapprox_gaussian} there exists a polynomial approximation $p(x)$ of the Gaussian function such that
    \begin{equation}
        |\pgauss{w/\alpha}{\eta} - e^{-\pi \alpha^2 x^2/(2w^2)}| \leq \, \eta\ \ \forall \, x \in [-1,1]
    \end{equation}
    where $\pgauss{w/\alpha}{\eta}$ has degree $\deg(\pgauss{w/\alpha}{\eta}) = \bigO{\frac{\alpha}{w}\log{1/\eta}}$.
    Note that
    \begin{equation}
         \rhop(\pgauss{w/\alpha}{\eta}) = \sum_i \frac{(\pgauss{w/\alpha}{\eta}(\lambda_i/\alpha))^2}{\sum_j (\pgauss{w/\alpha}{\eta}(\lambda_i/\alpha))^2} \ketbra{\lambda_i}{\lambda_i}, \,\, \rhoMC(H,0) = \sum_i \frac{e^{-\pi \lambda_i^2/w^2}}{\sum_j e^{-\pi \lambda_j^2/w^2}} \ketbra{\lambda_i}{\lambda_i}
    \end{equation}
    Let $S \coloneqq \sum_i e^{-\pi \lambda_i^2/w^2}$ and $\Delta \coloneqq \sum_i (p(\lambda_i/\alpha))^2 - e^{-\pi \lambda_i^2/w^2}$.
    Then,
    \begin{align} 
        \lvert (p(\lambda_i/\alpha))^2 - e^{-\pi \lambda_i^2/w^2} \rvert & \leq \eta \lvert p(\lambda_i/\alpha) + e^{-\pi \lambda_i^2/(2w^2)} \rvert\nonumber\\
        &\leq \eta(2+\eta) \nonumber\\
        \implies |\Delta| &\leq \eta(2+\eta)2^N.
    \end{align}
    Let $\epsilon_{\eta} \coloneqq \eta(2+\eta)$. We thus have
    \begin{align}
        \norm{\rhop(\pgauss{w/\alpha}{\eta})-\rhoMC(H,0)}_1 &= \norm{\sum_i\left(\frac{(p(\lambda_i/\alpha))^2}{S + \Delta} - \frac{e^{-\pi\lambda_i^2/w^2}}{S} \right) \ketbra{\lambda_i}{\lambda_i}}_1 \nonumber\\
        &= \sum_i \left|\frac{S(p(\lambda_i/\alpha))^2}{S(S + \Delta)} - \frac{(S+\Delta)e^{-\pi \lambda_i^2/w^2}}{S(S+\Delta)}  \right| \nonumber\\
        & \leq \sum_i \left(  \frac{\epsilon_{\eta}}{|S + \Delta|} + \left|\frac{\Delta e^{-\pi \lambda_i^2/w^2}}{S(S + \Delta)} \right|\right) \nonumber\\
        & \leq \frac{\epsilon_{\eta} 2^N + |\Delta|}{|S + \Delta|} \nonumber\\
        & \leq \frac{2\epsilon_{\eta} 2^N}{|S + \Delta|}.
    \end{align}
Since $\lambda_i/\alpha \in [-1,1], S \geq e^{-\pi \alpha^2/w^2} 2^N$. Therefore,
    \begin{equation}
    \label{eq:rho_approx}
        \norm{\rhop(\pgauss{w/\alpha}{\eta})-\rhoMC(H,0)}_1  \leq \frac{2\epsilon_{\eta}}{e^{-\pi\alpha^2/w^2} - \epsilon_{\eta}}, \\
    \end{equation}
    Selecting $\eta \leq \frac{\epsilon_p e^{-\pi \alpha^2/w^2}}{9} < 1$, we have
    \begin{align}
        &\epsilon_{\eta} \leq \frac{\epsilon e^{-\pi \alpha^2/w^2}}{3} \nonumber\\
        \implies &\norm{\rhop(\pgauss{w/\alpha}{\eta})-\rhoMC(H,0)}_1 \leq \epsilon_p.
    \end{align}
    The degree of $p$ is $\bigO{\frac{\alpha}{w}\log{1/\eta}} = \bigO{\frac{\alpha}{w}(\log{1/\epsilon} + \pi \alpha^2/w^2)}$, which completes the proof.
\end{proof}    
\begin{lemma}
\label{lem:approx_poly_block_encoding}
    Let $H$ be a $k$-local Hamiltonian on $N$ qubits, $\alpha \leq \norm{H}$, and let the eigenvectors and eigenvalues of $H$ be $\{\ket{\lambda_i}, \lambda_i\}$. For any $\epsilon_p \in (0,1]$, let $\pgauss{w/\alpha}{\eta}$ be chosen as in \cref{lem:polyapprox_microcanonical}. Let $U$ be an $(N+m)$-qubit, $(\alpha, m, \epsilon')$-block encoding of $\pgauss{w/\alpha}{\eta}(H/\alpha)$. 
     Then, a single application of $U$ (postselecting on the correct block-encoding outcome) can be used to prepare a state $\rho'$ such that
    \begin{equation}
        \norm{\rho' - \rhoMC(H,0)}_1 \leq \sqrt{8\epsilon'} + \epsilon_p.
    \end{equation}
\end{lemma}
\begin{proof}
   As in \cref{lem:polyapprox_microcanonical}, let $\ket{\psi} \coloneqq \sum_i \frac{1}{\sqrt{2^N}} \ket{\lambda_i}_a\ket{\lambda_i}_b$, where $N$ qudits form the main computational register $a$, $N$ qudits form an auxiliary register $b$, and $\ket{\lambda_i}$ are the eigenstates of $H/\alpha$. $\ket{\psi}$ is the maximally entangled state represented in the basis of eigenstates of $H/\alpha$. We also define a register $c$ for the $n_c$ block encoding ancilla qudits. $U$ acts on all three of these registers and by definition of the block-encoding,
    \begin{equation}
        \norm{\bra{0^{n_c}}_c U \ket{0^{n_c}}_c - p(H/\alpha)} \leq \epsilon'.
    \end{equation}
    Let
    \begin{equation}
        \ket{\psi_p} \coloneqq  \frac{p(H/\alpha)_a\ket{\psi}}{\norm{p(H/\alpha)_a\ket{\psi}}}.
    \end{equation}
    Then, as in \cref{prop:poly_blockencoding}
    \begin{equation}
        \rhop = \ptr{b}{\ketbra{\psi_p}{\psi_p}}.
    \end{equation}
    Let $M \coloneqq \bra{0^{n_c}}_c U \ket{0^{n_c}}_c$. Then, let the state obtained after postselecting on the correct block-encoding being performed be
    \begin{equation}
        \ket{\phi} = \frac{M\ket{\psi}}{\norm{M\ket{\psi}}}.
    \end{equation}
    Further, let
    \begin{equation}
        \rho' \coloneqq \ptr{b}{\ketbra{\phi}{\phi}}.
    \end{equation}
    Then, by Proposition 9 of \textcite{Childs_Kothari_Somma_2017},
    \begin{equation}
        \norm{\ket{\psi_p} - \ket{\phi}} \leq \epsilon',
    \end{equation}
    By monotonicity of the trace distance under the partial trace,
    \begin{equation}
        \norm{\ptr{b}{\ketbra{\phi}{\phi}} - \ptr{b}{\ketbra{\psi_p}{\psi_p}}}_1 \leq \norm{ \ketbra{\phi}{\phi} - \ketbra{\psi_p}{\psi_p}}_1.
    \end{equation}
    Using this, and the relation between the fidelity and the trace distance, we obtain
    \begin{align}
        \norm{\rho' - \rhop}_1 &\leq 2\sqrt{2\epsilon'-\epsilon'^2} \nonumber\\
        &\leq \sqrt{8\epsilon'}.
    \end{align}
    Combining this with \cref{lem:polyapprox_microcanonical} using the triangle inequality, we obtain
    \begin{equation}
        \norm{\rho' - \rhoMC(H,0)}_1 \leq \sqrt{8\epsilon'} + \epsilon_p.
    \end{equation}
    \end{proof}
    
    Finally, we have the required ingredients to prepare the microcanonical state. 
    \begin{lemma}
\label{lem:rhoMC_sampling}
    For any $k$-local Hamiltonian $H$ on $N$ qubits, and energy $E$, there is a quantum algorithm using $\bigO{\poly{N, \log{1/\epsilon}}}$ space and $\bigO{\poly{N,\log{1/\epsilon}}\exp(((\norm{H} + |E|)/w)^2)}$ expected time that prepares a state $\rho$ such that
    \begin{equation}
        \norm{\rho - \rhoMC(H,E)}_1 \leq \epsilon.
    \end{equation}
\end{lemma}
\begin{proof}
Our algorithm is simple. Starting in the initial state $\ket{\psi} \coloneqq \frac{1}{\sqrt{2^s}} \sum_i  \ket{\lambda_i}_a\ket{\lambda_i}_b$, we apply a block encoding such that the resultant state, after postselecting on the correct block encoding being performed and tracing out the subsystem $b$, is a good approximation of the microcanonical state. It remains to be shown that such a block encoding can be constructed, and to analyze the approximation error as well as the time and space requirements.

First, we show how the block encoding is constructed. Since each term of $H$ is $k$-local and there are $\bigO{\poly{n}}$ terms in $H$, $H' \coloneqq H-EI$ is $\bigO{2^k \poly{n}}$-sparse. We can select $\alpha \geq \norm{H} + |E|$ so that the eigenvalues of $H'/\alpha$ lie in $[-1,1]$. $\alpha = \poly{n}$, so by \cref{lemma:sparseblockenc}, $H'$ has an $(\poly{n}, \poly{n}, \epsilon_H)$-block-encoding that can be implemented by a circuit with $\bigO{\poly{n},\poly{\log{\frac{1}{\epsilon_H}}}}$ gates and $\bigO{\poly{n},\poly{\log{\frac{1}{\epsilon_H}}}}$ space.

We make use of the polynomial approximation for a Gaussian function $\pgauss{w/\alpha}{\eta}$ (\cref{lemma:polyapprox_gaussian}) with degree $\deg(\pgauss{w/\alpha}{\eta})$ which we will determine later. Along with \cref{lemma:polyblock}, this implies there is a circuit $\tilde{U}_\delta$ with $\bigO{\poly{n},\poly{\log{\frac{1}{\epsilon_H}}}}$ gates that implements a ${(1,\poly{n}, 4 \deg(\pgauss{w/\alpha}{\eta}) \sqrt{\epsilon_H/\alpha} + \delta)}$ block-encoding of $\pgauss{w/\alpha}{\eta}(H'/\alpha)$, and has a classical description computable in time $\bigO{\poly{n, \deg(\pgauss{w/\alpha}{\eta}),\log{1/\delta}}}$.

The above construction makes use of parameters $\eta, \epsilon_H, \delta$. We would like to choose these parameters such that we approximate the target state well, and the postselection probability is not too small.
For the first condition,
we apply \cref{lem:approx_poly_block_encoding} with $\epsilon' = 4 \deg(\pgauss{w/\alpha}{\eta}) \sqrt{\epsilon_H/\alpha} + \delta$, we have
    \begin{equation}
        \norm{\rho' - \rhoMC(H,0)}_1 \leq  \left(\sqrt{32 \deg(\pgauss{w/\alpha}{\eta}) \sqrt{\epsilon_H/\alpha} + 8\delta}\right)  + \epsilon_p.
    \end{equation}
    We will solve for parameters such that the following condition is fulfilled:
    \begin{equation}
    \label{eq:condition1}
        \norm{\rho' - \rhoMC}_1 \leq \epsilon.
    \end{equation}

For our second condition, we require that our algorithm doesn't take too long. The time taken is determined by the probability of correctly postselecting, which we now estimate. Let $M \coloneqq \bra{0^{n_c}}_c \tilde{U}_\delta \ket{0^{n_c}}_c$. The probability of postselecting to the correct state is $\norm{M\ket{\psi}}^2$. 
Observe that 
    \begin{align}
        \norm{M\ket{\psi}} &\geq \norm{\pgauss{w/\alpha}{\eta}(H'/\alpha)\ket{\psi}} - \norm{(M - \pgauss{w/\alpha}{\eta}(H'/\alpha))\ket{\psi}} \nonumber\\
        &\geq \norm{\pgauss{w/\alpha}{\eta}(H'/\alpha)\ket{\psi}} - 4 \deg(\pgauss{w/\alpha}{\eta})\sqrt{\epsilon_H/\alpha} - \delta \nonumber\\
        &\geq e^{-\pi \alpha^2/(2w^2)} - \eta - 4 \deg(\pgauss{w/\alpha}{\eta})\sqrt{\epsilon_H/\alpha} - \delta.
    \end{align}
    In order for the probability of correct postselection to be lower bounded, require that the probability of correct postselection remains of the order of its leading term, i.e.,
    \begin{equation}
    \label{eq:condition2}
        \norm{M\ket{\psi}} = \bigomega{\exp(-\pi \alpha^2/(2w^2))}.
    \end{equation}

    We now solve for $\eta, \epsilon_H, \delta$ such that \cref{eq:condition1} and \cref{eq:condition2} are fulfilled. 
    
    First, to partially satisfy \cref{eq:condition1}, we select $\epsilon_p = \epsilon/2$. This implies $\eta \leq \frac{\epsilon e^{-\pi \alpha^2/w^2}}{18}$. 
    \Cref{eq:condition1} is then satisfied if $\sqrt{8\delta} < \epsilon/4$ and $\sqrt{32\deg(\pgauss{w/\alpha}{\eta})\sqrt{\epsilon_H/\alpha}}< \epsilon/4$.
    
    Now, in order to satisfy \cref{eq:condition2}, we impose the constraint that $\eta, 4 \deg(\pgauss{w/\alpha}{\eta})\sqrt{\epsilon_H/\alpha}, \delta$ are each less than $\frac{1}{4} e^{-\pi \alpha^2/(2w^2)}$. By our choice of $\epsilon_p$, we have already obtained a sufficient condition for $\eta$.

    To simultaneously fulfill $\sqrt{8\delta} < \epsilon/4$ and $\delta < \frac{1}{4} e^{-\pi \alpha^2/(2w^2)}$, we select $\delta < \epsilon^2 e^{-\pi \alpha^2/(2w^2)} / 32$.

    We make use of the fact that $\deg(\pgauss{w/\alpha}{\eta}) = \gamma \left(\frac{\alpha}{w}(\log{2/\epsilon} + \pi \alpha^2/w^2) \right)$ for some constant $\gamma$. 
    To simultaneously fulfill $\sqrt{32\deg(\pgauss{w/\alpha}{\eta})\sqrt{\epsilon_H/\alpha}} < \epsilon/4$ and $4 \deg(\pgauss{w/\alpha}{\eta})\sqrt{\epsilon_H/\alpha} < \frac{1}{4} e^{-\pi \alpha^2/(2w^2)}$, we select
    \begin{equation}
        \epsilon_H \leq \alpha \left(\frac{\epsilon e^{-\pi \alpha^2/(2w^2)}}{128\gamma \left((\alpha/w)(\log{2/\epsilon} + \pi \alpha^2/w^2) \right)} \right)^4.
    \end{equation}

    With the chosen values of $\eta, \epsilon_H, \delta$, \cref{eq:condition1} and \cref{eq:condition2} are satisfied. We now analyze the time and space requirements of the algorithm. The probability of correct postselection is $\norm{M\ket{\psi}}^2$. Hence, in order to postselect onto the target state $\rho'$, we require an expected $\bigO{e^{\pi\alpha^2/w^2}}$ repetitions of the block encoding circuit. We do not need any additional space despite an expected total exponential depth since we are repeating the same circuit until correct postselection.
    Further, the depth of the block-encoding circuit as well as the classical computation time to obtain its description remain polynomial since $\log{\frac{1}{\epsilon_H}}, \log{\frac{1}{\delta}} = \bigO{\log{1/\epsilon}(\alpha/w)^2}$.

    The depth of the block-encoding circuit scales as
    \begin{equation}
        \bigO{\left(\frac{\alpha}{w}\right) \log{1/\eta} \times \poly{\log{1/\epsilon_H}}} = \bigO{\poly{N, \log{1/\epsilon}, 1/w}}.
    \end{equation}
    
    Therefore, to prepare $\rho$ we require $\bigO{e^{\pi\alpha^2/w^2}}$ expected executions of a $\bigO{\poly{N, \log{1/\epsilon}, 1/w}}$ depth circuit, which can be performed in $\bigO{\poly{N, \log{1/\epsilon}, 1/w}}$ space.
\end{proof}
This algorithm allows us to estimate observables of the microcanonical state.
\begin{corollary}
    \label{cor:microcanonical_observable}
    The algorithm in \cref{lem:rhoMC_sampling} can be used to estimate $\widehat{\AN}$ such that $|\widehat{\AN} - \tr{\AN \rhoMC}| \leq \epsilon$ using $\bigO{\poly{N, \log{1/\epsilon}}}$ space and $\bigO{\poly{N,\log{1/\epsilon}}\exp(((\norm{H} + |E|)/w)^2)}$ expected time.
\end{corollary}
\begin{proof}
    Note that by the triangle inequality, $\norm{A_N} \leq \max_i \norm{A_i} = c$ for some constant $c$. Let $\rho$ be the state as prepared in \cref{lem:rhoMC_sampling}, to precision $\epsilon/c$.
    Applying Holder's inequality,
     \begin{align}
        \vert\tr{\AN\rho - \AN\rhoMC(H,E)}\vert \nonumber&\leq \Vert \AN(\rho-\rhoMC(H,E))\Vert_1\nonumber\\
        &\leq \Vert\rho-\rhoMC(H,E)\Vert_1\Vert \AN\Vert \nonumber\\
        &\leq \epsilon.
    \end{align}
    Hence by preparing $\rho$ and measuring $\AN$, we can estimate $\tr{\AN \rhoMC}$ to arbitrary constant precision. $\AN$ is a sum of $N$ single site (and hence commuting) observables. A single measurement of $\AN$ can be performed by preparing a sample of $\rho$ and measuring each of the $A_i$ observables sequentially, adding them, and dividing by $N$. By Hoeffding's inequality, to estimate $\tr{\AN \rhoMC}$ to precision $\epsilon$ with probability at least $2/3$ requires $\bigO{1/\epsilon^2}$ samples. Hence there is a a probabilistic polynomial space quantum algorithm to estimate $\tr{\AN \rhoMC}$ to constant precision. We note however, that this algorithm requires postselection, which may only succeed with inverse-exponentially small probability when $w = \bigO{\poly{N}}$.
\end{proof}
Finally, we can use this algorithm to prove one of our main results.
\begin{theorem}
    \label{thm:fstherm_pspace_containment}
    $\FSTherm \in \PSPACE$. 
\end{theorem}
\begin{proof}
Using the algorithm from \cref{cor:microcanonical_observable} along with the probabilistic polynomial space quantum algorithm to estimate $\barAN$ to constant precision \cref{lem:long_time_qpoly}, when $w = \bigomega{\frac{1}{\poly{N}}}$, $\FSTherm$ can be decided by a polynomial-space quantum algorithm that simply compares the two values, if our algorithm postselects correctly, which occurs with non-zero probability. On the other hand, in the event our algorithm does not postselect correctly, we can Accept or Reject with uniform probability. Hence, there is a quantum algorithm that such that for YES instances of $\FSTherm$, we accept with probability strictly greater than $1/2$, while for NO instances, we accept with probability strictly less than $1/2$. Hence, $\FSTherm \in \PrQSPACE$. 
    By \textcite{Watrous_1999}, $\PrQSPACE \subseteq \PSPACE$. Therefore, $\FSTherm \in \PSPACE$. 
\end{proof}

Another statistical ensemble of interest is the Gibbs ensemble $\rhog = e^{-\beta H}/\tr{e^{-\beta H}}$, representing the state of a system in contact with a bath of temperature $1/\beta$. An alternative to $\FSTherm$ is the problem of deciding whether a system thermalizes to the \textit{Gibbs} ensemble, i.e., whether the long-time average of an observable converges to the value corresponding to the Gibbs ensemble. Here, we note that the Gibbs state can be prepared probabilistically in polynomial quantum space, and hence this decision problem is also contained in $\PSPACE$. 
\begin{lemma}
    \label{thm:fstherm_gibbs_containment}
    There is a quantum algorithm using $\bigO{\poly{N ,\log{1/\epsilon}}}$ space and $\bigO{\exp(N)}$ time that computes $\widetilde{\AN}$ such that
    \begin{equation}
        \lvert \tr{\AN \rhog} - \widetilde{\AN} \rvert \leq \epsilon.
    \end{equation}
    Hence $\FSThermG \in \PSPACE$
\end{lemma}
\begin{proof}
    The Gibbs state $\rhog = \frac{e^{-\beta H}}{\tr{e^{-\beta H}}}$ can be $\epsilon$-approximated with a block-encoding of circuit depth $\bigO{\sqrt{\beta}\log{1/\epsilon}}$ and $\bigO{\sqrt{d^N/Z}}$ amplitude amplification steps \cite{gilyen2018qsvt,Apeldoorn_Gilyén_Gribling_Wolf_2020, Chowdhury_Somma_2016}. 
    Note that $Z = \tr{e^{-\beta H}} \geq d^N e^{-\beta \norm{H}} = \bigomega{1/e^{\poly{N}}}$, so estimating $\tr{\rhog \AN}$ to constant precision only requires $\bigO{\exp(N)}$ repetitions of a $\bigO{\poly{N}}$-depth circuit, which can be done in polynomial quantum space.
    Since $\PrQSPACE \subseteq \PSPACE$, $\FSThermG \in \PSPACE$.
\end{proof}

 \section{Hardness of \texorpdfstring{$\FSTherm$ and $\FSThermG$}{FSTherm(MC) and FSTherm(Gibbs)}}\label{sec:FSThermHardness}

In this subsection we will show that $\FSTherm$ is as hard as $\FSHalt$  under \textit{quantum} polynomial time reductions. We prove this by showing that given an instance of $\FSHalt$, we may use a poynomial time quantum algorithm to construct an instance of $\FSTherm$. We note that since we make use of a quantum polynomial time reduction rather than a classical one, our reduction differs from the standard formal complexity theoretic reduction.
We modify the $\FSRelax$ hardness construction such that we have $A^*=\tr{\rhoMC \AN},$ thus showing that there exist observables such that determining whether the system thermalizes is as hard as a $\PSPACE$-complete problem. We describe a high-level overview of our approach:
\begin{enumerate}
    \item We pad the end of the Turing machine even more in the case of halting, so that when it halts, $\barAN=(1+\bigO{1/N})(1-\alpha)\bra{a_2}A\ket{a_2}.$ This is covered in \cref{sec:boostingAN}.
    \item We double the size of the local Hilbert space, and modify $A_i$ by introducing parameters such that the two halves of the Hilbert space have different, tunable values of $\langle A_i\rangle$. We then modify the Hamiltonian so that, in the microcanonical state, both halves of the Hilbert space are represented equally. This construction allows us to tune the microcanonical expectation of $\AN$. This is covered in \cref{sec:changingquantum}.
    \item Finally, we show that a polynomial time quantum algorithm can be used to tune these parameters such that in the case of halting, $\barAN $ is equal to the microcanonical average of $\AN,$ and both are separated from $0,$ the non-halting value of $\AN.$ This is covered in \cref{sec:mc_qpoly_reduction}.
\end{enumerate}
Finally, in \cref{sec:fstherm_gibbs_hardness}, we prove that $\FSThermG$ is $\PSPACE$-hard (and hence $\PSPACE$-complete due to \cref{thm:fstherm_gibbs_containment}).
\subsection{Modifying the TM}\label{sec:boostingAN}
We modify the Turing Machine construction such that the values of $\barAN$ on halting inputs and on non-halting inputs are sufficiently separated, so that we may perform our tuning procedure in the following sections to reduce $\FSRelax$ to $\FSTherm$.
We first show how to asymptotically increase $\barAN$ in the case of halting to $(1-\alpha)\bra{a_2}A\ket{a_2}$, which we require for technical reasons in our tuning construction in \cref{sec:changingquantum}. 
As discussed in \cref{Sec:Turing_Machine_Description}, TM1 performs the actual computation of a URTM, TM2 flips A cells if TM1 halts, and TM3 runs TM1 in reverse. Since in the case of halting, the A cells are flipped  by TM2 and remain flipped for the duration of TM3, the long-time average of $\AN$ is determined by the duration of TM2 and TM3. We can thus modify this average by changing the time spent in TM3.
We achieve this by modifying TM3 to first run a buffer which performs a full loop around the tape for each element of the tape (i.e., taking $N(N+1)$ steps). Then, we continue as before with TM3 running TM1 in reverse, but with an additional buffer that iterates the finite control around the entire loop until it returns to its initial location in between each (computational) step of TM3.
Since a full loop around the tape takes time $N$, we thus increase the time spent in the computational part of TM3 by a factor of $N+1$.

To ensure the correct TM action, we add to the tape alphabet of TM3 such that every symbol in the tape alphabet has a ``place-marker'' copy which will be what the TM initially writes to the tape. 
This copy lets the finite-control know when it has returned to its most recent write cell. 
To facilitate the traversal, we will also add a copy of every finite-control state, which the update rules will just shift through the ring until it hits the ``place-marker'' tape-symbols.
For the first buffer, we need to add another state of the head corresponding to the specific buffer, and two symbols to the tape alphabet to help keep track of the progress of the $N$ loops around the tape. Specifically we will initialize the two symbols, pushing one of the symbols in one direction after every loop, and only finishing when the symbols come in contact again. 
Thus, in the case of halting, this additional looping increases the total run-time to $(N+2)T_h + (N+1)N$. The first buffer ensures that this run-time is large even when $T_h < N$.

\begin{lemma}
	\label{lem:halting_separation_modified}
	With the above modified Turing Machine construction, If $y$ encodes a Turing Machine that either halts in time $T_h$ or never halts, then there exists $\delta \in (0,N^{-1}]$ such that
	\begin{align}
		\barAN&=\begin{cases}
			(1-\alpha)\left[1-\delta\right]\bra{a_2} A \ket{a_2} & \text{halting}\\
			0 & \text{non-halting}.
		\end{cases}
	\end{align}
\end{lemma}
\begin{proof}
	The proof proceeds along exactly the same lines as the proof of \cref{lem:halting_separation}.
Revisiting \cref{eq:expvaluegeneral}, with the modification that $T=(N+2)T_h + (N+1)N,$ we have that
\begin{align}
    \barAN &= (1-\alpha)\frac{T-(T_h+N/2+1/2)}{T+1} \nonumber\\
    &=(1-\alpha) \frac{(N+2)T_h + (N+1)N- (T_h+N/2+1/2)}{(N+2)T_h + (N+1)N+1} \bra{a_2} A \ket{a_2}\nonumber\\
    &= (1-\alpha)\left[1-\frac{T_h+N/2+3/2}{(N+2)T_h + (N+1)N+1}\right]\bra{a_2} A \ket{a_2}\nonumber\\
    &= (1-\alpha)\left[1-\delta\right]\bra{a_2} A \ket{a_2},
\end{align}
where $\delta \coloneqq \frac{T_h+N/2+3/2}{(N+2)T_h + (N+1)N+1},$ and $ 0 < \delta \leq 2N^{-1}.$ 
The non-halting case will still have $\barAN=0.$ 
\end{proof}

\subsection{Modifying \texorpdfstring{$H,$ $\mathcal{H},$ and $A$}{H,H, and A}}\label{sec:changingquantum}

In addition to this modification to the Turing Machine, we will also modify the Hilbert space and Hamiltonian. The goal of this modification is to:\begin{itemize}
    \item Push the microcanonical expectation value of $\AN$ to the middle of $A$'s spectrum.
    \item Give us control of $\barAN$ in the halting case in such a way that we are guaranteed that it crosses the microcanonical expectation away from the non-halting case's value (zero).
\end{itemize}
To do these, we will replace our local Hilbert space $\mathcal{H}$ with a new Hilbert space $\mathcal{H}_{\text{tune}}:=\mathcal{H}\oplus \mathcal{H}$ of doubled dimension. 
We will denote the projection onto $\mathcal{H}\oplus \mathbf{0}$ by $\Pi$ and the projection onto $\mathbf{0}\oplus \mathcal{H}$ by $\Pi'.$ Denote the basis of $\mathcal{H}\oplus \mathbf{0}$ by $\{\ket{a}\ :\ a \in Q\cup \Gamma\}$ and that of $\mathbf{0}\oplus\mathcal{H}$ by $\{\ket{a'}\ :\ a \in Q\cup \Gamma\}.$ Let 
\begin{equation}
    H_{TM} \coloneqq \sum_{abcd\in\mathcal{T}} \ketbra{cd}{ab}+h.c,
\end{equation} represent Turing-machine evolution on $\mathcal{H}\oplus \mathbf{0},$ and
\begin{equation}
H_{TM}' \coloneqq \sum_{abcd\in\mathcal{T}} \ketbra{c'd'}{a'b'}+h.c,    
\end{equation}
represent Turing machine evolution on $\mathbf{0}\oplus\mathcal{H}.$ We then let $H\coloneqq H_{TM}+H'_{TM}.$ Note that $[H_{TM},H'_{TM}]=0.$ We also define 
\begin{equation}
    F\coloneqq \sum_{a \in Q \cup \Gamma} \ketbra{a'}{a} + h.c.
\end{equation}
Next, we slightly modify $A_i.$
Let $A_i  \coloneqq p\ketbra{a_2}{a_2}_i + q\Pi'_i$, where $p,q$ are parameters we will choose. 
With these modifications, we compute the expectation of $\AN$ for the microcanonical ensemble and Gibbs ensembles as special cases of any ensemble which is analytic in $H$.
\begin{lemma}
\label{lem:analytic_expectation}
Consider any $\rho \propto f(H)$ where $f$ is analytic and $\rho$ is unit trace. Then,
    \begin{equation}
        \tr{\mathcal{A}_N \rho} = \frac{q}{2} + p \tr{\ketbra{a_2}{a_2}_1 \rho}.
    \end{equation}
\end{lemma}
\begin{proof}

First, note that $[F,H]=0.$ As $\rho\propto f(H)$ for some analytic function $f:\mathbb{R}\to\mathbb{R},$ it follows that $[\rhoMC,F]=0$ by the existence of a power series expansion of $f;$ we note that this holds for arbitrary functions $f$ whose domain contains the spectrum of $H,$ but this is a bit more involved to prove.
Therefore, as $F^2=I,$
\begin{align}
    \tr{\Pi'_i \rho}&= \tr{\Pi'_i F\rho F}\nonumber\\
    &= \tr{F\Pi'_i F\rho}\nonumber\\
    &= \tr{\Pi_i \rho}\nonumber\\
    &= \tr{(I-\Pi'_i) \rho}.
\end{align}
Thus, as $\rho$ is unit-trace, $\tr{\Pi'_i \rho}=1/2.$ 
Hence, we have
\begin{align}
        \tr{\mathcal{A}_N \rho} &=\frac{p}{N} \sum_i \left(\tr{ \ketbra{a_2}{a_2}_i \rho} + q\tr{\Pi'_i\rho}\right) \nonumber\\
        &= \frac{q}{2} + p \tr{\ketbra{a_2}{a_2}_1 \rho}. 
\end{align}
where we have used the fact that $\rho$ is translationally invariant since $H$ is translationally invariant, which follows from the same logic as showing that $[\rho,F]=0$ above.
\end{proof}

Applying \cref{lem:analytic_expectation} to the microcanonical state we have
\begin{corollary}
\label{lem:microcanonical_expectation}
    \begin{equation}
        \tr{\mathcal{A}_N \rhoMC} = \frac{q}{2} + p \tr{\ketbra{a_2}{a_2}_1 \rhoMC}.
    \end{equation}
\end{corollary}
Likewise, applying \cref{lem:analytic_expectation} to the Gibbs state we have
\begin{corollary}
\label{lem:gibbs_expectation}
    \begin{equation}
        \tr{\mathcal{A}_N \rhog} = \frac{q}{2} + p \tr{\ketbra{a_2}{a_2}_1 \rhog}.
    \end{equation}
\end{corollary}

\subsection{A quantum polynomial time reduction from \texorpdfstring{$\FSHalt$ to $\FSTherm$}{FSHalt to FSTherm(MC)}}
\label{sec:mc_qpoly_reduction}
\begin{figure}[h!]
    \centering
\includegraphics[width=0.5\textwidth]{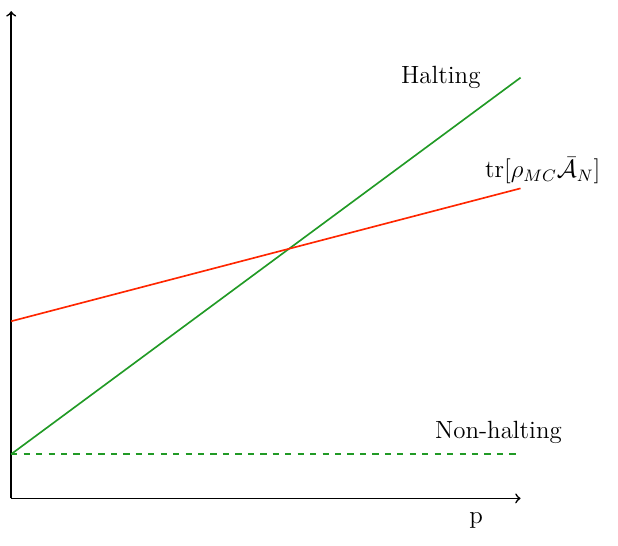}
    \caption{The solid green and dashed-green lines represent the time-averaged expectation values in the halting and non-halting cases respectively, as $p$ is tuned. 
    The red line represents the expectation values with respect to the microcanonical ensemble as $p$ is tuned. 
    }
    \label{fig:Tuned_Expectation_Value}
\end{figure}
We now show how $p$ can be chosen so that the long time average of $\AN$, i.e., $\barAN$, in the halting case and the expectation of $\AN$ for the microcanonical ensemble both agree. We briefly comment here on some constraints of this construction. Our reduction relies on a quantum algorithm to prepare the microcanonical state, but the algorithm presented in \cref{sec:fstherm_containment} in general requires $\bigO{e^{\pi\alpha^2/w^2}}$ expected executions of a $\bigO{\poly{N, \log{1/\epsilon}, 1/w}}$ depth circuit, where $\alpha$ upper bounds $\norm{H} + |E|$. We therefore show hardness for a restricted class of microcanonical states with  width $w = \bigomega{\alpha/\sqrt{\log{N}}}$. Such a window width is still vanishingly small with respect to the range of the Hamiltonian spectrum, but allows for efficient preparation of the microcanonical state.

\begin{lemma}
\label{lem:mc_qpoly_reduction}
    Let $H$ be a $k$-local Hamiltonian with $cN$ terms, each with norm bounded by 1. For energy $E$, let $\rhoMC(H,E)$ be the microcanonical state with window width $w = \bigomega{cN/\log{N}}$. Then, for any $\epsilon$ there exists $N_0$ such that for all $N\geq N_0$ there is a quantum algorithm that takes time $\bigO{\poly{N, \log{\frac{1}{\epsilon}}} }$ to compute $p$ such that with the above Turing machine construction and observable $\AN$, 
    \begin{equation}
        \lvert\tr{\rhoMC\AN} - \barAN\rvert \leq \epsilon
    \end{equation}
\end{lemma}
\begin{proof}
    In the above sections using \cref{lem:halting_separation_modified} and \cref{lem:microcanonical_expectation}, we showed how an instance of $\FSHalt$ can be used to construct an instance of $\FSRelax$ with parameters $p,q$ such that:
\begin{enumerate}
    \item In the case of non-halting, $\barAN = 0$
    \item In the case of halting, $\barAN= (1-\alpha)\left[1-\delta\right]p.$, with $\delta < 1/N$.
    \item $\tr{\rhoMC\AN} = \frac{q}{2} + \frac{p\tr{\rhoMC \ketbra{a_2}{a_2}}}{2}$
\end{enumerate}
In the case of halting, to tune $p$ such that $\barAN \approx \tr{\rhoMC\AN}$, we require that $q$ is chosen such that $\tr{\rhoMC\AN} > \barAN$ at $p=0$, and $\tr{\rhoMC\AN} < \barAN$ at $p=1$.
Note that as $\Pi-\ket{a_2}\!\bra{a_2}\succeq 0,$ it follows that $0\leq\tr{ \ketbra{a_2}{a_2}_i\rhoMC}\leq \tr{\Pi_i\rhoMC}\leq 1/2.$ Using the fact that $\tr{\rhoMC\ketbra{a_2}{a_2}_1} \leq 1/2$, the above constraint on $q$ is reduced to
\begin{equation}
    \frac{q}{2} \leq (1-\alpha)[1-\delta] - \frac{1}{4}.
\end{equation}
when $q = \frac{1}{2}$, for any $N$ there exists a sufficiently small constant $\alpha$  such that this constraint is satisfied. Therefore, in the rest of this section we assume $q=\frac{1}{2}$.
Using the polynomial time quantum algorithm described in \cref{lem:rhoMC_sampling}, for constant $w$, we can estimate $\widetilde{\AN}$ such that $\lvert \tr{\rhoMC \ketbra{a_2}{a_2}_1} - \widetilde{\AN} \rvert \leq \epsilon/2$ for constant error $\epsilon/2$. 
To ensure that in the case of halting, $\barAN \approx \tr{\rhoMC\AN}$, we would like to select an appropriate $p$, corresponding the point of intersection of the 2 values in \cref{fig:Tuned_Expectation_Value}. With $p = \frac{1}{4(1-\alpha) - 2\widetilde{\AN}}$, we obtain
\begin{equation}\label{eq:tuning_approx_error}
    \lvert\tr{\rhoMC\AN} - \barAN\rvert \leq \left\lvert \frac{(1-\alpha)\delta}{4(1-\alpha) - 2\widetilde{A_N}} \right\rvert + \frac{\epsilon}{4}
\end{equation}
Since $\delta \leq N^{-1}$, there exists $N_0$ such that for all $N \geq N_0$, $\lvert\tr{\rhoMC\AN} - \barAN\rvert \leq \epsilon$.
\end{proof}

We note that as we have doubled our Hilbert space dimension twice from the relaxation construction, $d=204.$ in the above construction. Combining \cref{thm:fsrelax_pspace_hardness} and \cref{lem:mc_qpoly_reduction}, we obtain the following theorem.
\begin{theorem}
\label{thm:fstherm_pspace_hardness}
    For any instance of $\FSHalt$, there exists a Hamiltonian $H$, input state $\ket{\psi_0}$, and sum of local observables $\AN$ for which the long time average of $\AN$ equals $\tr{\AN\rhoMC}$ in the case of halting and equals 0 in the case of non-halting. Further, $H, \ket{\psi_0}, \AN$ can be constructed by a polynomial time quantum algorithm from instances of $\FSHalt$. Hence, $\FSTherm$ is $\PSPACE$-hard under a quantum polynomial time reduction.
\end{theorem}

Our main result \cref{thm:fstherm_pspace_complete} follows from \cref{thm:fstherm_pspace_hardness} and \cref{thm:fstherm_pspace_containment}.

\subsection{A classical polynomial time reduction from \texorpdfstring{$\FSHalt$ to $\FSThermG$}{FSHalt to FSTherm(Gibbs)}}
\label{sec:fstherm_gibbs_hardness}
Using similar techniques to \cref{sec:mc_qpoly_reduction}, we will prove a \emph{classical} polynomial time reduction from $\FSHalt$ reduces to $\FSThermG$.
The only notable difference from \cref{sec:changingquantum} is that we must compute appropriate $\beta$ such that $\tr{H\rhog} = E \coloneqq \bra{\psi_0} H \ket{\psi_0}$ efficiently classically, and then compute $\tr{\ketbra{a_2}{a_2} \rhog}$ efficiently classically in order to correctly tune the value of $p$ (using \cref{lem:gibbs_expectation}). We can compute both of these using the  result \cite{Alhambra_Cirac_2021} that there is a classical polynomial time algorithm to estimate local observables of Gibbs states with finite $\beta$ in 1D.
However, in order to do this, we must know in advance limits on the value of $|\beta|$. In order to make use of these algorithms we require that $\beta$ is constant, which we ensure in \cref{sec:gibbs_modification}.

In order to estimate $\beta$, we can make use of the fact that $E(\beta) = \tr{H\rhog}$ is monotonic in $\beta$. If $\beta$ is guaranteed to be $\bigO{1}$, we can use the algorithm from \cite{Alhambra_Cirac_2021} to estimate $\tr{H\rhog}$ and iteratively improve our estimate of $\beta$ to error $\epsilon$ in $\log{1/\epsilon}$ steps classically using a binary search algorithm.

Finally, if $\beta = \bigO{1}$, we may efficiently compute the expectation value of $\tr{\ketbra{a_2}{a_2} \rhog}$, once again using the result of \textcite{Alhambra_Cirac_2021}.

How do we guarantee that $\beta = \bigO{1}$? Observe that as $\beta$ approaches $-\infty$ or $\infty$, $\rhog$ approaches the highest excited state or the ground state respectively. Hence, we expect that there is a range of expected energies for which $\beta$ is constant, which are found away from the edges of the spectrum of $H$. If $E$ is extensively gapped from the edges of the spectrum, in the thermodynamic limit (i.e., large $N$ limit) the energy density must be a constant distance away from the edges of the spectrum as well. This would imply that the corresponding $\beta$ must be finite in the thermodynamic limit (as it is only infinite at the edges of the spectrum), and hence in the finite-sized limit we expect it to be constant. We formalize this in \cref{lem:constantbeta}.
This Lemma implies that for a Hamiltonian $H$ used in our hardness construction, if the initial state is guaranteed to be in the middle of the spectrum, the corresponding Gibbs state has $\beta = \bigO{1}$.  By simple modifications of $H$ discussed in the following section (increasing the local Hilbert space dimension and applying local energy penalties), it is possible to ensure that the initial product state encoding hard instances of $\FSHalt$ has expected energy extensively separated from the edges of the spectrum. Hence, there is a classical polynomial time reduction from $\FSHalt$ to $\FSThermG$, and $\FSThermG$ is $\PSPACE$-complete.\subsection{Modifying the Hilbert Space.}
\label{sec:gibbs_modification}
We first modify $\mathcal{H}$ to push $E(\beta)$ away from the edges of the spectrum. To do this, we use a similar technique to \cref{sec:changingquantum}, augmenting the Hilbert space with two new sectors that duplicate the behaviour of the original sector, just with energies shifted appropriately up and down. In particular, we first triple the local Hilbert space to 
\begin{equation}
\mathcal{H}_{\text{Gibbs}} = \mathcal{H}\oplus\mathcal{H}\oplus\mathcal{H}.
\end{equation}
We further modify our Hamiltonian by adding a local term on each site that applies a penalty of $-\epsilon$ for any basis state of $\mathcal{H}\oplus \mathbf{0} \oplus \mathbf{0}$, and a penalty of $\epsilon$ for any basis state of $\mathbf{0} \oplus \mathbf{0} \oplus \mathcal{H}$. 
For any eigenstate of $\widetilde{H}$ with energy $E$ and support only in the middle sector ($\mathbf{0} \oplus \mathcal{H} \oplus \mathbf{0}$) of $\mathcal{H}_{\text{Gibbs}}$, there thus exists a corresponding eigenstate with support only in the upper sector ($\mathbf{0} \oplus \mathbf{0} \oplus \mathcal{H}$) with energy $E + \epsilon N$ and a corresponding eigenstate with support only in the lower sector ($\mathcal{H} \oplus \mathbf{0} \oplus \mathbf{0}$) with energy $E - \epsilon N$.
Hence, any eigenstate with support only in the middle sector of $\mathcal{H}_{\text{Gibbs}}$ must have energy gapped away from the edges of the spectrum of $\widetilde{H}$ by at least $\epsilon N$.
 Finally, we let our observable $\AN':=\mathbf{0}\oplus \AN\oplus \mathbf{0}$, i.e., it is only supported in the middle sector. 

 So far, by tripling the Hilbert space and modifying the Hamiltonian appropriately, we have ensured that our computational states are extensively separated from the edges of the spectrum. Now, in order to tune our observable as before, we perform the same type of Hilbert space doubleing and Hamiltonian modification as in \cref{sec:changingquantum}.
 Hence, our final local Hilbert space is $\mathcal{H}_{\text{tune}}:=\mathcal{H}_{\text{Gibbs}}\oplus \mathcal{H}_{\text{Gibbs}}$. If $\widetilde{H}$ represents the Turing machine evolution (tripled and gapped) on the space $\mathcal{H}_{\text{Gibbs}}\oplus \mathbf{0}$, and $\widetilde{H'}$ represents the Turing machine evolution on $\mathbf{0} \oplus \mathcal{H}_{\text{Gibbs}}$, we make use of the Hamiltonian $H = \widetilde{H} + \widetilde{H'}$.
We further slightly modify $A_i$, again following \cref{sec:changingquantum}.
Let $A_i  \coloneqq p\ketbra{a_2}{a_2}_i + q\Pi'_i$, where $p,q$ are parameters we will choose, and $\ket{a_2}$ is in the middle sector of the space $\mathcal{H}_{\text{Gibbs}}\oplus \mathbf{0}$, while $\Pi'$ is a projector onto $\mathbf{0} \oplus \mathcal{H}_{\text{Gibbs}}$. Then, the analysis proceedds as in \cref{sec:changingquantum}, and we can make use of \cref{lem:gibbs_expectation}.
Since we have first tripled, and then doubled our local Hilbert space, the local Hilbert space dimension is six times that of our $\FSRelax$ construction.

Thus, with our newly modified Hamiltonian, observable, and Hilbert space, we are able to use the following lemma.
\begin{lemma}\label{lem:constantbeta}
    For any 1D nearest neighbor Hamiltonian $H$, and product state $\ket{\psi}$, let $E = \bra{\psi} H \ket{\psi}$. Let the maximum and minimum eigenvalues of $H$ be $E_{max}, E_{min}$ respectively.
    If there exists a constant $\epsilon$ such that $E_{min} + \epsilon N \leq E \leq E_{max}-\epsilon N$, and if $\beta$ is such that $\tr{\rhog H} = E$, then $\beta = \bigO{1}$
\end{lemma}
\begin{proof}
    For each eigenvalue $E_i$, let $\Delta_i \coloneqq E_i - E_{min}$.
    Then, let
    \begin{equation}
        S(\beta) \coloneqq \sum_{i} e^{-\beta \Delta_i}
    \end{equation}
    and $p$ be the distribution 
    \begin{equation}
        p_i(\beta) \coloneqq \frac{e^{-\beta \Delta_i}}{S(\beta)}.
    \end{equation}
    Observe that $E(\beta) - E_{min} = \sum_{i} p_i(\beta) \Delta_i$. Now, consider the Shannon entropy of $p$:
    \begin{align}
        H(p) &= -\sum_i p_i \log{p_i} \nonumber \\
        &= \sum_i e^{-\beta \Delta_i} \beta(E_i - E_{min}) \nonumber \\
        &= \beta(E(\beta) - E_{min}) 
    \end{align}
    Now, observe that $H(p)$ is bounded by the maximum entropy in $d^N$ dimensions. Hence $H(p) \leq N\log{d}$.
    Since $E(\beta) - E_{min} \geq \epsilon N$, we have
    \begin{align}
        \beta \epsilon N &\leq N \log{d} \nonumber \\
        \implies \beta &\leq \frac{\log{d}}{\epsilon}.
    \end{align}
    Performing a similar analysis for the other edge of the spectrum, $E_{max}$, we find that if 
    \begin{equation}
        E_{max} - E(\beta) \geq \epsilon N,
    \end{equation}
    then
    \begin{equation}
        |\beta| \leq \frac{\log{d}}{\epsilon}.
    \end{equation}
    Since $\epsilon, d$ are constant, $\beta = \bigO{1}$.
\end{proof}

Before we prove our result, we need one more lemma about the use of the bisection method for root-finding in the presence of error in the estimation of the function.
\begin{lemma}\label{lemma:noisybisec}
Let $f(x)$ be a strictly monotonic function with zero $f(x^*)=0$, and let $\hat{f}(x)$ be some approximation to $f(x)$ satisfying $\sup_{x}\vert f(x)-\hat{f}(x)\vert \leq \delta.$ 
Suppose that $x_\downarrow< x^* < x_\uparrow,$ such that $-f(x_\downarrow),f(x_\uparrow)>\delta.$ 
Now, suppose further that $\vert f'(x)\vert > m$ for all $x\in [x_\downarrow,x_\uparrow]$. 
Then, if we apply the bisection method to $\hat{f},$ with starting points $x_{\downarrow},x_{\uparrow}$ and running it until either $\lceil \mathrm{log}_2 \frac{x_{\uparrow}-x_{\downarrow}}{\varepsilon}\rceil$ steps have elapsed or we obtain some value $\hat{x}$ with $\vert f(\hat{x})\vert\leq \delta,$ 
we will get an approximate zero $\hat{x}$ with $\vert \hat{x}-x^*\vert\leq \varepsilon,$ so long as $\delta < m\varepsilon. $ 
\end{lemma}
\begin{proof}
    Let $I:=\{x\ \vert\ \vert{f}(x)\vert\leq \delta\}.$ Then, $\text{sign}\hat{f}(x)=\text{sign}f(x)$ for $x\notin I.$ Thus, if $I\subseteq [x^*-\varepsilon, x^*+\varepsilon],$ the bisection method will run as if it has access to $f,$ not $\hat{f},$ 
    up until the final step. To guarantee that $I\subseteq [x^*-\varepsilon,x+\varepsilon],$ note that
\begin{align*}
\min\{\vert f(x^*-\varepsilon) \vert, \vert f(x^*+\varepsilon)\vert\} &\geq \varepsilon\inf_{u\in [x^*-\varepsilon, x^*+\varepsilon]} \vert f'(u)\vert \\
&\geq m\varepsilon.
\end{align*}
Thus, if $\delta < m\varepsilon,$ we have that $I\subset [x^*-\varepsilon,x^*+\varepsilon].$
\end{proof}

Finally, we may now prove our result.
\begin{theorem}
\label{thm:fstherm_gibbs_pspace_hardness}
    For any instance of $\FSHalt$, there exists a Hamiltonian $H$, input state $\ket{\psi_0}$, and sum of local observables $\AN$ for which the long time average of $\AN$ equals $\tr{\AN \rhog}$ in the case of halting and equals 0 in the case of non-halting. Further, $H, \ket{\psi_0}, \AN$ can be constructed by a polynomial time classical algorithm from instances of $\FSHalt$. Hence, $\FSThermG$ is $\PSPACE$-hard.
\end{theorem}
\begin{proof}
    We make use of the modified Hilbert Space and Hamiltonian construction to encode an instance of $\FSHalt$ into the initial state, with Turing machine evolution performed by $H$. 
    The initial state is encoded into the middle sector of $\mathcal{H}_{\text{Gibbs}}\oplus \mathbf{0}$, and hence obeys the conditions required of \cref{lem:constantbeta}. 
    Hence, $\beta$ of the Gibbs state corresponding to the energy of the initial state is a constant. 
    Note that for any constant $\beta$, $u(\beta):= \tr{\rhog H}/N,$ the energy density, can be estimated to inverse polynomial error by a polynomial time classical algorithm \cite{Alhambra_Cirac_2021}. 
    In particular, let us estimate $u(\beta)$ to $\mathcal{O}(N^{-2})$ precision. 
    Then, to apply \cref{lemma:noisybisec}, we need only show that for $N$ sufficiently large, $u'(\beta)>m>0$ for the relevant range of $\beta$. 
    To see this, we note that by \cite{araki_gibbs_1969}
    , the energy density is analytic in the thermodynamic limit, and by \cite{Simon1993,griffiths_strict_1971}
    , we have that $u'(\beta)>0$ for all $\beta$ in the thermodynamic limit, meaning that there exists some $m,N_0$ such that on any compact set of $\beta$s, for all $N>N_0,$ $u'(\beta)>m>0.$ Thus, if we estimate $u(\beta)$ to $\mathcal{O}(N^{-1})$ precision, we will get $\beta$ to $\mathcal{O}(N^{-1})$ precision as well. This is sufficient to estimate $\tr{\AN\rhog}$ to constant precision, because \begin{align}
    \Vert \rhog-\rhogp\Vert_1 = \mathcal{O}(N\vert \beta-\beta'\vert),
\end{align} 
where the prefactors only depend on the details of the local interaction of Hamiltonian. 
Thus, by the matrix H\"{o}lder inequality, we see that if $\beta'-\beta = \mathcal{O}(N^{-1}),$ the error in the expectations of $\AN$ is $\bigO{1}.$

    
    Having estimated $\beta$, we may once again make use of the protocol of ~\textcite{Alhambra_Cirac_2021} to estimate local observables of 1D Gibbs states in order to estimate $\tr{\AN\rhog},$ or alternatively $\tr{\ket{a_2}\bra{a_2}_1\rhog},$ efficiently. Thus, we can efficiently perform the same tuning procedure as in \cref{lem:mc_qpoly_reduction}, and specifically by setting $p$ as discussed above \cref{eq:tuning_approx_error}, to tune the Gibbs expectation to the halting value of $\barAN.$ Hence, $\FSHalt$ is reduced to $\FSThermG$.
\end{proof}
Our main result \cref{thm:fsthermg_pspace_contained} follows from \cref{thm:fstherm_gibbs_containment} and \cref{thm:fstherm_gibbs_pspace_hardness}.

\section{BQP-Hardness of \texorpdfstring{$\FTFSRelax$}{FTFSRelax}}\label{sec:FTFSRelax}

\subsection{Modifying the TM-to-Hamiltonian Construction} \label{Sec:QTM_Intro}
We return to $\FTFSRelax$ \cref{def:ftfsrelax} showing that it is $\BQP$-hard following a very similar proof to the $\PSPACE$-hardness proof of $\FSRelax$. 
However, instead of encoding the problem $\FSHalt$, we instead choose to encode a $\BQP$-complete computation.
The only major change this requires is that TM1 is no longer a classical Turing Machine, but is replaced with a quantum Turing Machine (QTM).
A QTM functions similarly to a classical TM, but where the TM head and tape configuration can be in a superposition of states \cite{bernsteinqct}.
At each time-step of the QTM evolution, the tape configuration and control head update is described by a unitary $U$ which is determined by the QTM transition rules.
If the initial QTM configuration is $\ket{\psi}$, the state at the next time-step is $U\ket{\psi}$.
A more detailed description is given in \cite{bernsteinqct}.

Similar to the $\PSPACE$-hardness proof, we run TM2 and TM3 only if the computation outputs an accepting state.
However, in the rejecting state, we choose to run a TM2', which is identical to TM2 except it does not flip $a_1\rightarrow a_2$, but otherwise exactly mimics the motion of TM2.
The analysis of the eigenvalues and eigenstates of the Hamiltonian proceeds identically: following \cref{Sec:TM-to_Hamiltonian}.
Since the encoded TMs are guaranteed to halt, we find that the eigenvalues of the effective Hamiltonian are given as $\lambda_j(\Heff) = 2\cos\left(\frac{2\pi j}{T}\right)$.

\subsection{Time-Averaged Expectation Values}

In order to see the long-time relaxation behaviour of the Hamiltonian, we need to observe the Hamiltonian for a sufficiently long time.

The only other difference is that the $\BQP$ computation will only accept with probability $>2/3$ or $<1/3$.
Thus, we use a standard amplification procedure to push the acceptance probabilities to $1-2^{-\poly{N}}$ and $2^{-\poly{N}}$ respectively \cite{marriott2005quantum}.
Using the same definition of $\ket{w_k}$ as \cref{Sec:TM-to_Hamiltonian}, we see that 
\begin{align}
    \bra{w_j}\AN \ket{w_j}  \begin{cases}
        = 0 \quad \quad   \quad  \quad   \quad \quad     &j<T_{halt}, \\ 
        > (1-\alpha)(1-2^{-\poly{N}}) \quad &j>T_{halt} + T_2 \  \text{and TM1 accepts,}  \\
        < (1-\alpha)2^{-\poly{N}}   &j>T_{halt} + T_2 \ \text{and TM1 rejects.} 
    \end{cases}
\end{align}
For the intermediate times, where TM2 is still acting, we note that in the rejecting case $\bra{w_j}\AN \ket{w_j}<(1-\alpha)2^{-\poly{N}}$.

To simplify the analysis, we first note that provided we are averaging over sufficiently long times, the averaged expectation value is close to the infinite time value.
We prove the following in \cref{Sec:Finite_Time_Expectation_Proof}:
\begin{restatable}{lemma}{finitetimeexpvalue}\label{lemma:finite_time_exp_value}
    Consider the time-averaged expectation value:
    \begin{align}
        \barAN(\tau) = \frac{1}{\tau}\int_0^\tau dt   \bra{\psi_0}e^{iHt}\AN e^{-iHt}\ket{\psi_0}.
    \end{align}
    Then, the following bound between the infinite time and finite time expectation values holds:
    \begin{align}
        |\barAN(\tau) - \barAN| = \bigO*{ \frac{T^2}{\tau}}.
    \end{align}
\end{restatable}
\noindent Hence we simply need to choose $\tau\gg T^2$ and our analysis will reduce to the analysis in infinite time limit.
Since $T = \poly{N}$, for convenience we choose $\tau = p(N) T^2$, for some polynomial $p$, as an appropriate (but still polynomial) time period.
Repeating the analysis of \cref{Sec:Long-Time_Expectation_Values}, we find that 
\begin{align}
    \barAN(\tau) \begin{cases}
        >\frac{1}{2}(1-\alpha)(1-1/{\poly{N}}) -  \bigO*{ \frac{1}{p(N)} } \quad &\text{TM1 accepts,} \\
        < \bigO*{ \frac{1}{p(N)}} &\text{TM1 rejects.}
    \end{cases}
\end{align}
\noindent We can then verify that the promise is satisfied for this Hamiltonian provided in an identical way to \cref{Lemma:Promise_Satisfied}, but noting that $T=\bigO{\poly{N}}$.

\begin{theorem}
\label{thm:FTFSRelax_BQP_Complete}
    $\FTFSRelax$ is $\BQP$-complete.
\end{theorem}
\begin{proof}
    By the above analysis, the value of $\barAN(\tau)$ can be made to depend on whether an instance of a $\BQP$-complete problem is accepted or rejected, so $\FTFSRelax$ is $\BQP$-hard. Further, due to the promise, we can estimate $\barAN$ to constant precision by simply simulating Hamiltonian evolution for polynomial time and measuring $\AN$. Hence $\FTFSRelax \in \BQP$, which implies $\FTFSRelax$ is $\BQP$-complete.
\end{proof}
\section{Discussion}
\label{sec:discussion}
We have shown there exist Hamiltonians for which, given an initial state and an observable to measure, predicting the long-time average of the observable is computationally intractable even for a quantum computer.
An immediate consequence of this is that, assuming $\PSPACE$ does not collapse to a lower (easy) complexity class, there is no ``easy'' way of generically determining whether a given Hamiltonian and initial state thermalizes.
This should be compared to the eigenstate thermalization hypothesis, which states that Hamiltonians satisfying certain conditions on their matrix elements should thermalize.
Hence, either checking whether a Hamiltonian satisfies the conditions of the ETH on its matrix elements must not be computationally easy, or the ETH must not be a full characterization of all Hamiltonians which thermalize.

We briefly discuss here the similarities and differences between our construction and that of \textcite{shiraishithermal2021}. Primarily, their work focuses on systems in the thermodynamic limit (i.e. $N \rightarrow \infty)$, while we focus on finite sized systems.  While their result applies to almost-uniform product state inputs, our hardness construction is for non-uniform product state inputs in order to satisfy the promise of inverse-exponentially-separated eigenvalues for the eigenvectors with support on the input state. We take advantage of this relaxation to change the Turing machines used, as we do not need to decode the input to the URTM. We also make use of a more concise local Hilbert space for our hardness proof of $\FSRelax$, although we currently require greater local dimension to prove the hardness of $\FSTherm$.
Finally, as their result shows undecidability, they did not need to prove containment for their decision problem. 

Concurrently and independently of this work, Matsumoto \cite{Matsumoto_2025} further generalized the result of \textcite{shiraishithermal2021}, showing intractability of thermalization in the thermodynamic limit even for almost i.i.d. input states. Matsumoto also showed that in a finite-sized lattice, a relaxation problem defined similarly to ours is either $\PSPACE$-complete or $\EXPSPACE$-complete with a modified definition of the input size.

Our contribution in this work also includes $\PSPACE$ algorithms to compute the long time observable average, as well as the microcanonical expectation values, and thereby decide $\FSTherm$. 

There are several limitations to our work. 
In our tuning construction, we showed hardness for a restricted family of microcanonical states, with energy windows having inverse-logarithmic width ratio with the Hamiltonian norm. This is a technical restriction which arises due to the complexity of preparing the microcanonical state. In addition, in our proof of $\PSPACE$-containment of $\FSTherm$, we assumed that $w = \bigomega{\frac{1}{\poly{N}}}$. In future work, we hope to show $\PSPACE$-completeness for the more general microcanonical state, with narrower energy windows.
Our tuning construction also requires large local Hilbert space dimension. We hope to refine this result in future versions of this work. A hardness result with low local dimension would make our results much more applicable to `natural' Hamiltonians.

\paragraph{Further Work}
While in \cref{sec:containment} we showed that $\FSTherm$ is contained in $\PSPACE$ for both the Gibbs ensemble and the microcanonical ensemble with greater than inverse-polynomial width of the energy window, our hardness proof in \cref{sec:FSThermHardness} only applies to the microcanonical ensemble with an energy window $w = \bigomega{\norm{H}/\sqrt{\log{N}}}$. Further, the hardness proof relies on a quantum polynomial time reduction as opposed to a more standard classical polynomial time reduction. In future work, we hope to show hardness under the more standard reduction, as well as for most general ensembles. 
In \cref{sec:fstherm_gibbs_hardness}, we proved that $\FSThermG$ is $\PSPACE$-hard. Thus, combined with \cref{thm:fstherm_gibbs_containment}, $\FSThermG$ is $\PSPACE$-complete (and using only a classical polynomial time reduction).
Our argument in \cref{sec:fstherm_gibbs_hardness} that there is a classical polynomial time reduction from $\FSHalt$ to $\FSThermG$ makes use of the fact that in 1D for constant $\beta$, there exists a polynomial time classical algorithm to compute Gibbs expectation values \cite{Alhambra_Cirac_2021}. This algorithms make use of the locality of $H$ and exponential decay of correlations to construct matrix product operators for Gibbs states. In our tuning construction in \cref{sec:changingquantum} to prove hardness of $\FSTherm$, which relies on computing local observables, we do not make use of the additional structure of $H$ when estimating local observables of $\rhoMC$. One possible strategy to make the reduction from $\FSHalt$ to $\FSTherm$ classical rather than quantum could be to make use of this structure to show that local observables of $\rhoMC$ can be estimated efficiently classically. Such an approach may show that $\FSTherm$ is indeed $\PSPACE$-complete. Another strategy would be to use ensemble equivalence results, such as those of Brandão and Cramer \cite{brandao2015equivalence}, to provide quantiative bounds on the deviation of the microcanonical and Gibbs expected values for sufficiently large system size, which we can then use to efficiently classically perform the tuning procedure for microcanonical ensembles.

One direction for future work is investigating complexity of different formulations of thermalization and equilibration. An example of a notion of thermalization and ensembles that is not captured by our work is equilibration of integrable systems to generalized Gibbs ensembles \cite{Rigol_Dunjko_Yurovsky_Olshanii_2007, DAlessio_Rigol_2016}. An additional open question is related to thermalization of subsystems. For example, in a system with two parts $A,B$, with dynamics generated by a hamiltonian $H = H_A + H_B + H_{AB}$, we may like to know under which conditions and inverse temperatures $\beta$ the behavior of subsystem $A$ is captured by the Gibbs ensemble $\propto e^{-\beta H_A}$.

An alternative route to further research is to consider the reasons why Hamiltonians do not thermalize from a complexity-theoretic perspective.
Some relevant thermalization-inhibiting phenomena include many-body localization (MBL) and quantum scars. 
While the complexity of simulating systems with MBL been studied previously \cite{huang2015efficient, ehrenberg2022simulation}, the authors of this work do not know of any results concerning the complexity of determining MBL from a description of the Hamiltonian.

Two additional obvious and interesting paths are as follows: finding constructions with smaller local Hilbert space dimensions, and proving a complexity result without the assumption on the gap between eigenvalues. Recent work has shown the presence of ultraslow relaxation in systems based on the complexity of word problems, with small local dimension \cite{Balasubramanian_2024}. In future work we hope to explore whether similar techniques can be used to yield results on the hardness of thermalization but with smaller local dimension.

\section*{Acknowledgements}

	{\begingroup
		\hypersetup{urlcolor=navyblue}

We thank \href{https://orcid.org/0000-0002-3321-3198}{Brayden Ware}, \href{https://orcid.org/0000-0002-9969-7391}{Alex Schuckert}, \href{https://orcid.org/0000-0002-9903-837X}{Andrew Childs}, \href{https://orcid.org/0000-0003-0509-3421}{Alexey Gorshkov}, \href{https://orcid.org/0000-0002-5889-4022}{Alvaro Alhambra},
\href{https://orcid.org/0000-0001-7899-6619}{Stephen Piddock},
\href{https://orcid.org/0000-0002-9992-3379}{Sevag Gharibian},
\href{https://orcid.org/0000-0002-7918-2841}{Greeshma Oruganti}
and \href{https://orcid.org/0000-0003-3383-1946}{Joe Iosue}
for helpful discussions. We thank GPT-5 for suggesting the proof of \cref{lem:constantbeta} and for helpful discussions.

DD acknowledges support by the NSF GRFP under Grant No.~DGE-1840340 and an LPS Quantum Graduate Fellowship.
T.C.M.~and D.D.~were supported in part by the DoE ASCR Quantum Testbed
Pathfinder program (awards No.~DE-SC0019040 and No.~DE-SC0024220).
D.D.~was also supported in part by DARPA SAVaNT ADVENT. T.C.M.~and
D.D.~also acknowledge support from the U.S.~Department of Energy,
Office of Science, Accelerated Research in Quantum Computing,
Fundamental Algorithmic Research toward Quantum Utility (FAR-Qu).
J.D.W.~acknowledges support from the United States Department of Energy, Office of Science, Office of Advanced Scientific Computing Research, Accelerated Research in Quantum Computing program, and also NSF QLCI grant OMA-2120757.

	\endgroup}

{\begingroup
		\hypersetup{urlcolor=navyblue}
\printbibliography
	\endgroup}
 \appendix

\section{Finite Time-Averaged Expectation Values} \label{Sec:Finite_Time_Expectation_Proof}
In this appendix, we prove \cref{lemma:finite_time_exp_value}. First, let us recall the lemma:
\finitetimeexpvalue*

\begin{proof}

 In the finite time case, we can no longer use the integration property of $e^{-i(\lambda_i-\lambda_j)t}$. 
Instead we write:
\begin{align}
    \barAN(\tau) &= \frac{1}{\tau}\int_0^\tau dt \sum_{ij}e^{-i(\lambda_i-\lambda_j)}c_i^*c_j\bra{\lambda_i}\AN\ket{\lambda_j} \nonumber\\
    &= \sum_{ij}\delta(\lambda_i-\lambda_j)c_i^*c_j\bra{\lambda_i}A\ket{\lambda_j} + \frac{1}{\tau}\sum_{ij}\int_{[\tau_{ij}]}^\tau dt e^{-i(\lambda_i-\lambda_j)}c_i^*c_j\bra{\lambda_i}\AN\ket{\lambda_j},
\end{align}
where for a particular $i,j$ we define $[\tau_{ij}]$ as the largest multiple of the period of $e^{-i(\lambda_i-\lambda_j)}$ which is less than $\tau$.
We have used that $e^{-i(\lambda_i-\lambda_j)}$ integrates to zero if the integration is done over a multiple of the oscillatory period.
We now want to bound the second term.
Using the fact there are at most $T$ eigenstates, we get:
\begin{align}
     \frac{1}{\tau}\sum_{ij}\int_{[\tau_{ij}]}^\tau dt e^{-i(\lambda_i-\lambda_j)}c_i^*c_j\bra{\lambda_i}\AN\ket{\lambda_j} &\leq \frac{\norm{A}}{\tau}\int_{[\tau_{ij}]}^\tau dt e^{-i(\lambda_i-\lambda_j)}c_i^*c_j\bra{\lambda_i}\ket{\lambda_j} \nonumber\\
     &= \frac{\norm{A}}{\tau}\int_{[\tau_{ij}]}^\tau dt \bra{\psi(t)}\ket{\psi(t)} \nonumber\\
&\leq \frac{\norm{A}}{\tau}(\tau-[\tau_{ij}])    
\end{align}
We now note that $(\tau-[\tau_{ij}])$ must be less than a single period of the exponential $e^{-i(\lambda_i-\lambda_j)t}$, and hence $(\tau-[\tau_{ij}])\leq \max_{ij} 1/(\lambda_j-\lambda_i)$, where $\max_{ij} 1/(\lambda_j-\lambda_i)$ is the maximum period of any of the complex exponentials.
From \cref{Sec:QTM_Intro}, we see that 
\begin{align}
    \lambda_k-\lambda_m =   \cos\left( \frac{k\pi}{T+1} \right) -  \cos\left( \frac{m\pi}{T+1} \right).
\end{align}
The minimum will be achieved for $m=k+1$, hence:
\begin{align}
    |\lambda_k-\lambda_{k+1}| &=   \cos\left( \frac{k\pi}{T+1} \right) -  \cos\left( \frac{(k+1)\pi}{T+1} \right) \nonumber\\
    & \leq \frac{1}{2}\left( \frac{\pi}{T+1} \right)^2|[k^2 - (k+1)^2]| + \bigO*{\frac{1}{T^4} } \nonumber\\
    &=\frac{1}{2}\left( \frac{\pi}{T+1} \right)^2(2k+1)+ \bigO*{\frac{1}{T^4} }
\end{align}
The minimum is achieved for $k=0$.
Hence:
\begin{align}
     \frac{1}{\tau}\sum_{ij}\int_{[\tau_{ij}]}^\tau dt e^{-i(\lambda_i-\lambda_j)}c_i^*c_j\bra{\lambda_i}\AN\ket{\lambda_j} &\leq \frac{\norm{A}}{\tau}(\tau-[\tau_{ij}])    \nonumber\\
     &\leq \bigO*{\frac{T^2}{\tau} }.
\end{align}
    
\end{proof}

\end{document}

%% file: symphony.tex
\usepackage[capitalise]{cleveref}
\Crefname{lemma}{Lemma}{Lemmas}
\Crefname{proposition}{Proposition}{Propositions}
\Crefname{definition}{Definition}{Definitions}
\Crefname{theorem}{Theorem}{Theorems}
\Crefname{conjecture}{Conjecture}{Conjectures}
\Crefname{corollary}{Corollary}{Corollaries}
\Crefname{example}{Example}{Examples}
\Crefname{section}{Section}{Sections}
\Crefname{appendix}{Appendix}{Appendices}
\Crefname{figure}{Fig.}{Figs.}
\Crefname{equation}{Eq.}{Eqs.}
\Crefname{table}{Table}{Tables}
\Crefname{item}{Property}{Properties}
\Crefname{remark}{Remark}{Remarks}

\newtheorem{theorem}{Theorem}[section]
\newtheorem{definition}[theorem]{Definition}
\newtheorem{corollary}[theorem]{Corollary}
\newtheorem{proposition}[theorem]{Proposition}

\newtheorem{lemma}[theorem]{Lemma}

\usepackage{ltablex}

\newcommand\prob\textsc
\makeatletter
\newcommand{\probleminput}[1]{\gdef\@probleminput{#1}}
\newcommand{\problemquestion}[1]{\gdef\@problemquestion{#1}}
\newcommand{\problempromise}[1]{\gdef\@problempromise{#1}}
\DeclareDocumentEnvironment{problem}{}{
\probleminput{}\problempromise{}\problemquestion{}
\leavevmode
}{
  \par\addvspace{0\baselineskip}
  \noindent
  \begin{itemize}
      \item[\textbf{Input:}] \@probleminput
\ifx\@problempromise\empty%
\else%
      \item[\textbf{Promise:}] \@problempromise
\fi%
      \item[\textbf{Output:}] \@problemquestion
  \end{itemize}
}
\makeatother